\PassOptionsToPackage{unicode}{hyperref}
\PassOptionsToPackage{hyphens}{url}
\PassOptionsToPackage{dvipsnames,svgnames,x11names}{xcolor}
\documentclass[
  12pt]{article}

\usepackage{amsmath,amssymb,amsthm}
\usepackage{iftex}
\ifPDFTeX
  \usepackage[T1]{fontenc}
  \usepackage[utf8]{inputenc}
  \usepackage{textcomp} 
\else 
  \usepackage{unicode-math}
  \defaultfontfeatures{Scale=MatchLowercase}
  \defaultfontfeatures[\rmfamily]{Ligatures=TeX,Scale=1}
\fi
\usepackage{lmodern}
\ifPDFTeX\else  
\fi
\IfFileExists{upquote.sty}{\usepackage{upquote}}{}
\IfFileExists{microtype.sty}{
  \usepackage[]{microtype}
  \UseMicrotypeSet[protrusion]{basicmath} 
}{}
\makeatletter
\@ifundefined{KOMAClassName}{
  \IfFileExists{parskip.sty}{%
    \usepackage{parskip}
  }{
    \setlength{\parindent}{0pt}
    \setlength{\parskip}{6pt plus 2pt minus 1pt}}
}{
  \KOMAoptions{parskip=half}}
\makeatother
\usepackage{xcolor}
\makeatletter
\ifx\paragraph\undefined\else
  \let\oldparagraph\paragraph
  \renewcommand{\paragraph}{
    \@ifstar
      \xxxParagraphStar
      \xxxParagraphNoStar
  }
  \newcommand{\xxxParagraphStar}[1]{\oldparagraph*{#1}\mbox{}}
  \newcommand{\xxxParagraphNoStar}[1]{\oldparagraph{#1}\mbox{}}
\fi
\ifx\subparagraph\undefined\else
  \let\oldsubparagraph\subparagraph
  \renewcommand{\subparagraph}{
    \@ifstar
      \xxxSubParagraphStar
      \xxxSubParagraphNoStar
  }
  \newcommand{\xxxSubParagraphStar}[1]{\oldsubparagraph*{#1}\mbox{}}
  \newcommand{\xxxSubParagraphNoStar}[1]{\oldsubparagraph{#1}\mbox{}}
\fi
\makeatother

\usepackage{longtable,booktabs,array}
\usepackage{calc} 
\usepackage{etoolbox}
\makeatletter
\patchcmd\longtable{\par}{\if@noskipsec\mbox{}\fi\par}{}{}
\makeatother
\IfFileExists{footnotehyper.sty}{\usepackage{footnotehyper}}{\usepackage{footnote}}
\makesavenoteenv{longtable}
\usepackage{graphicx}
\makeatletter
\def\maxwidth{\ifdim\Gin@nat@width>\linewidth\linewidth\else\Gin@nat@width\fi}
\def\maxheight{\ifdim\Gin@nat@height>\textheight\textheight\else\Gin@nat@height\fi}
\makeatother
\setkeys{Gin}{width=\maxwidth,height=\maxheight,keepaspectratio}
\makeatletter
\def\fps@figure{htbp}
\makeatother

\makeatletter
\@ifpackageloaded{caption}{}{\usepackage{caption}}
\AtBeginDocument{%
\ifdefined\contentsname
  \renewcommand*\contentsname{Table of contents}
\else
  \newcommand\contentsname{Table of contents}
\fi
\ifdefined\listfigurename
  \renewcommand*\listfigurename{List of Figures}
\else
  \newcommand\listfigurename{List of Figures}
\fi
\ifdefined\listtablename
  \renewcommand*\listtablename{List of Tables}
\else
  \newcommand\listtablename{List of Tables}
\fi
\ifdefined\figurename
  \renewcommand*\figurename{Figure}
\else
  \newcommand\figurename{Figure}
\fi
\ifdefined\tablename
  \renewcommand*\tablename{Table}
\else
  \newcommand\tablename{Table}
\fi
}
\@ifpackageloaded{float}{}{\usepackage{float}}
\floatstyle{ruled}
\@ifundefined{c@chapter}{\newfloat{codelisting}{h}{lop}}{\newfloat{codelisting}{h}{lop}[chapter]}
\floatname{codelisting}{Listing}

\makeatother
\makeatletter
\@ifpackageloaded{caption}{}{\usepackage{caption}}
\@ifpackageloaded{subcaption}{}{\usepackage{subcaption}}
\makeatother

\ifLuaTeX
  \usepackage{selnolig}  
\fi
\usepackage[]{natbib}
\usepackage{bookmark}

\IfFileExists{xurl.sty}{\usepackage{xurl}}{} 
\hypersetup{
  pdftitle={Title},
  pdfauthor={Author 1; Author 2},
  pdfkeywords={3 to 6 keywords, that do not appear in the title},
  colorlinks=true,
  linkcolor={blue},
  filecolor={Maroon},
  citecolor={Blue},
  urlcolor={Blue},
  pdfcreator={LaTeX via pandoc}}

\usepackage{mathtools}
\usepackage{bbm,bm}
\usepackage{booktabs}
\usepackage{subdepth}
\usepackage{url}
\usepackage{ulem}
\usepackage{setspace}
\allowdisplaybreaks
\graphicspath{{Fig/}}
\usepackage{ulem}

\newcommand{\RR}{\mathbb{R}}
\newcommand{\vc}[1]{\bm{#1}}
\newcommand{\diff}{\, \textnormal{d}}
\newcommand{\PP}{\mathbb{P}}
\newcommand{\EE}{\mathbb{E}}

\newcommand{\summ}{\textnormal{sum}}
\newcommand{\dto}{\overset{\textnormal{d}}{\rightarrow}}

\newcommand{\vto}{\overset{v}{\rightarrow}}
\newcommand{\norm}[1]{\lVert {#1} \rVert}

\newcommand{\dm}{d}
\newcommand{\fac}{q}
\newcommand{\vX}{\vc{X}}
\newcommand{\vY}{\vc{Y}}
\newcommand{\vx}{\vc{x}}
\newcommand{\vy}{\vc{y}}
\newcommand{\siid}{\overset{\textnormal{iid}}{\sim}}

\newtheorem{theorem}{Theorem}[section]
\newtheorem{prop}[theorem]{Proposition}%
\newtheorem{lem}[theorem]{Lemma}%

\theoremstyle{definition}

\newtheorem{exa}[theorem]{Example}%
\newtheorem{rem}[theorem]{Remark}%

\graphicspath{{Fig/}}

\newcommand{\anon}{1}

\DeclareRobustCommand{\change}[1]{{\color{black}#1}}

\begin{document}

\def\spacingset#1{\renewcommand{\baselinestretch}%
{#1}\small\normalsize} \spacingset{1}


\if1\anon
{
  \title{\bf Estimating probabilities of multivariate failure sets based on pairwise tail dependence coefficients}
  \author{Anna Kiriliouk \thanks{
    Anna Kiriliouk gratefully acknowledges funding from the Fonds de la Recherche Scientifique - FNRS (grant number J.0124.20)}\hspace{.2cm}\\
    LIDAM/ISBA, UCLouvain \\
    and \\
    Chen Zhou \\
    Econometric Institute, Erasmus School of Economics, \\
    Erasmus University Rotterdam}
  \maketitle
} \fi

\if0\anon
{
  \bigskip
  \bigskip
  \bigskip
  \begin{center}
    {\LARGE\bf Title}
\end{center}
  \medskip
} \fi

\bigskip
\begin{abstract}
Estimating probabilities of extreme events involving multiple risk factors is a critical challenge in fields such as finance and climate science. This paper proposes a  parametric approach to estimate the probability that a multivariate random vector falls into an extreme failure set, based on the information in the tail pairwise dependence matrix (TPDM) only. The TPDM provides a summary of tail dependence for all pairs of components of the random vector. We propose an efficient algorithm to obtain approximate completely positive decompositions of the TPDM, enabling the construction of a max-linear model whose TPDM approximates that of the original random vector. We also provide conditions under which the approximation turns out to be exact. Based on the decompositions, we can construct max-linear random vectors to estimate failure probabilities, exploiting their computational simplicity. We apply the proposed method to estimate probabilities of extreme events for real-world datasets, including industry portfolio returns and maximal wind speeds, demonstrating its practical utility for risk assessment.\end{abstract}

\noindent%
\vfill

\newpage
\spacingset{1.8} 

\section{Introduction}

\change{
Extreme value statistics deal with the characterization of extreme events, such as stock market crashes, that occur with low frequency but with a potentially severe impact. A typical applied question is to estimate the (low) probability of such an extreme event. When handling compound extreme events, we are interested in estimating the probability that a $d$-dimensional random vector $\vc{X}$ falls into some ``extreme region'' $C$ (also called \emph{failure region}). For example, in the context of portfolio risk, the components of $\vc{X}$ represent negative log-returns of $d$ stocks in a portfolio. Another example is in catastrophic climate events, where $\vc{X}$ may represent maximum daily rainfall amounts at $d$ monitoring stations. To estimate such a probability $\PP[\vc{X} \in C]$, it is important to take the tail dependence of $\vX$ into account, i.e., the tendency of extremes to occur in several components simultaneously. 

In this study, we particularly focus on failure regions containing few or even no observed data points. This lack of observations makes it challenging to estimate the associated failure probability non-parametrically.
One potential solution is to employ appropriate parametric models for multivariate extremes and to estimate the model parameters using intermediate level observations. However, as the dimension of $\vc{X}$ increases, such a parametric approach becomes increasingly complex, which leads to the following dilemma. On the one hand, models with a small number of parameters tend to oversimplify the tail dependence structure; on the other hand, models with many parameters are difficult to estimate, especially without additional information (e.g. geographical locations of weather stations).

This paper provides a parametric model for the case where the dimension $d$ of the underlying random vector $\vX$ is moderate to high. The model can capture the so-called \emph{tail dependence coefficients} for all pairs of components of $\vX$. Despite the large number of parameters (up to $d(d+1)/2$), estimation is based on a simple and fast algorithm, allowing one to approximate the failure probability using the estimated model. The parametric model we use can be summarized via bivariate dependence coefficients, similar to covariances in the (centered) multivariate normal distribution. 

In the extreme-value literature, tail dependence of $\vc{X}$ is characterized in various ways: by the exponent measure, the angular measure, and the stable tail dependence function, among others \citep{beirlant2004, dehaanferreira2006}.  While such objects fully characterize the extremal behavior of a random vector, as the dimension increases, their estimation suffers from the well-known \emph{curse of dimensionality}. Instead,
one may consider a partial summary of tail dependence by focusing on  the tail dependence coefficients $\chi$ or the extremal coefficients $\theta$ for all pairs, see e.g. \citet{coles1999}. This is similar to summarizing the dependence structure of a random vector by the covariance or correlation matrix. Most existing parametric approaches for multivariate extremes also use pairwise tail dependence coefficients as model diagnostics and for assessing goodness-of-fit; see e.g., \citet{huser2022advances} and \citet{hentschel2025statistical}.

In this paper, we focus on a generalization of the so-called \emph{tail pairwise dependence matrix} (TPDM), which collects all pairwise tail dependence coefficients. The TPDM was originally proposed by \citet{larsson2012} in the bivariate case, under the assumption of multivariate regular variation, and has gained in popularity since the work of \citet{cooley2019}; for example, \citet{kluppelberg2019} use it to estimate the parameters of an extreme Bayesian network, \citet{fomichov2020} propose a clustering algorithm inspired by spherical $k$-means, where the means are replaced by the Perron-Frobenius eigenvectors of the TPDM, and \citet{kim2022} study the pairwise tail dependence between vectors of scores of functional data. \change{We show that our generalized definition has desirable invariance properties, and is applicable for regularly varying vectors with different tail indices. In addition, our TPDM can be directly related to the so-called \emph{angular measure}, extending its applicability beyond the standard multivariate regularly varying context to the broader class of distributions belonging to the max-domain of attraction of multivariate generalized extreme value distributions.}  

The tail pairwise dependence matrix is closely related to the \emph{max-linear model}, a simple but rather flexible multivariate extreme-value model. In such a model, each component of a $d$-dimensional vector can be interpreted as the maximum shock among a set of $q$ independent heavy-tailed factors. The max-linear model has $dq$ parameters (or $d(q-1)$ when margins are standardized), usually stocked in a parameter matrix $A \in \RR^{d \times q}$. 
Max-linear models are, among other applications, used to model extremes on directed acyclic graphs \citep{gissibl2015} or as latent factor structures for financial returns \citep{cui2018}. If the number of parameters is limited, one can use minimum distance estimation  \citep{einmahl2018} or spherical $k$-means \citep{janssen2020} to estimate a max-linear model. \citet{kiriliouk2020} presents hypothesis tests aiding in deciding the size of $q$. Theoretically, \citet{fougeres2013} proved that the max-linear model is dense in the class of $d$-dimensional multivariate extreme-value distributions. 

\citet{cooley2019} showed  that there exists a finite $q \in \mathbb{N}$ such that the tail pairwise dependence matrix $\Sigma_{\vX}$ of any multivariate regularly varying random vector $\vc{X}$ with $\alpha = 2$ is equal to that of a max-linear model with $q$ factors via the relation $\Sigma_{\vX}=AA^T$. 
In practice, the parameter matrix $A$ of the max-linear model can be obtained through a completely positive decomposition of $\Sigma_{\vX}$. 
A matrix $S$ is called \emph{completely positive} if there exists a non-negative matrix $A$ such that $S = A A^T$; finding this matrix $A$ amounts to finding a \emph{completely positive decomposition} of $S$. Note that the completely positive decomposition is not unique, i.e., multiple max-linear models may share the same TPDM. Positive decompositions are well studied in the linear algebra literature \citep{barioli2003,groetzner2020,bot2021}, where it is shown that the maximum number of columns of $A$ is $d(d+1)/2 - 4$. Algorithms exists to obtain positive decompositions, but they are computationally intensive.

In this paper, we propose an algorithm to obtain an \emph{approximate} completely positive decomposition $\Sigma_{\vX} \approx A A^T$, where $A$ is a $(d \times d)$-dimensional matrix with non-negative elements. In other words, we aim at choosing $q=d$, which is in line with classical principal component analysis. The algorithm is easy to compute and applicable for large $d$. We show that the decomposition is exact when $\vc{X}$ follows a max-linear model with a triangular coefficient matrix $A$. 
When the decomposition is not exact, it can still be used to construct a max-linear model whose TPDM matches $\Sigma_{\vX}$ for all off-diagonal entries but possibly overestimates its diagonal values.  
In that case, for a certain class of failure regions, the corresponding probability of a failure set based on the constructed model can serve as a conservative estimate for the one based on the original random vector $\vX$. In addition, the constructed max-linear model is simple to simulate from, allowing for efficient estimation of failure probabilities via Monte Carlo methods whenever analytical expressions are not available. The simulation-based method can also be used when the underlying distribution is not multivariate regularly varying, but belongs to the max-domain of attraction of multivariate generalized extreme value distributions.

Other estimators of multivariate failure probabilities are proposed in \citet{drees2015} and \citet{valk2016}. While these estimators are based on fewer assumptions than ours, their computational complexity does not allow estimation of failure probabilities for, say, $d \geq 5$. Other related literature \citep{cai2011,einmahl2013,he2017} focuses on the semi- or non-parametric estimation of multivariate extreme quantile regions (also limited to low dimensions), thus employing a narrower definition of a failure region. To the best of our knowledge, a general approach for estimating $\PP[\vc{X} \in C]$ where the dimension of $\vc{X}$ is moderate to high and $C$ is user-defined does not exist to date.

The outline of the paper is as follows. Section~\ref{sec:background} introduces our generalization of the TPDM and shows its theoretical properties. Section~\ref{sec:alg} presents the decomposition algorithm and assesses its accuracy in approximating failure probabilities. Statistical estimation of the TPDM and the associated failure probabilities is discussed and numerically assessed in Section~\ref{sec:estim}. We apply the proposed decomposition algorithm to real data of industry portfolio returns and maximal wind speeds in Section~\ref{sec:appl}. Finally, Section~\ref{sec:disc} concludes. The supplement contains all proofs and an additional application for heavy rainfall in Switzerland.}

\section{The tail pairwise dependence matrix}\label{sec:background}

We start by introducing our generalization of the tail pairwise dependence matrix (TPDM) for a multivariate regularly varying random vector in Section \ref{sec:mrv}. Then, in Section \ref{sec:maxlin}, we present the max-linear model and its relation to the TPDM. In Section \ref{sec:failure}, we discuss how the probability of a failure region can be approximated under a max-linear model. Finally, in Section \ref{sec:general}, we discuss how to go beyond multivariate regular variation, assuming only a general multivariate domain of attraction condition. 

\subsection{Tail dependence and multivariate regular variation}\label{sec:mrv}
\change{Let $\vc{X} = (X_1, \ldots,X_d)^T$ be a random vector on $[0,\infty)^d$} with joint cumulative distribution function $F$ and continuous marginal distributions $F_1, \ldots, F_d$. 
Write $\EE_0 = [0,\infty]^{\dm} \setminus \{\vc{0}\}$.
We \change{start by assuming} that $\vc{X}$ is \emph{multivariate regularly varying}, i.e., there exists a sequence $b_n \rightarrow \infty$ and a {non-degenerate} limit measure $\nu_{\vc{X}}$ such that
\begin{equation}\label{eq:regvar}
n \, \PP \left[ b_n^{-1} \vc{X}  \in \,\cdot \, \right] \vto \nu_{\vX} (\, \cdot \,), \qquad \text{as } n \rightarrow \infty,
\end{equation}
where $\vto$ denotes vague convergence in the space of non-negative Radon measures on $\EE_0$ and $\nu_{\vX}$ is called the \emph{exponent measure} \citep[Chapter 6]{resnick2007}.  Then there exists an $\alpha > 0$, called the \emph{tail index} of $\vX$, such that the
limit measure $\nu_{\vX}$ satisfies the homogeneity property of order $\alpha$, i.e., $\nu_{\vX}  (tC) = t^{-\alpha} \nu_{\vX} (C)$ for any constant $t > 0$ and any ($\nu_{\vX}$-measurable) Borel set $C \subset \EE_0$. The sequence $b_n$ can be shown to satisfy $b_n \sim L(n) n^{1/\alpha}$, where $L$ denotes a slowly varying function. 

The homogeneity property suggests that it is possible to express $\vc{X}$ and correspondingly the limit relation in \eqref{eq:regvar} in polar coordinates.
Let $\norm{\, \cdot \,}$ denote any norm on $\RR^d$ and consider the transformation $T: 	\EE_0 \mapsto (0, \infty) \times \mathbb{S}_{\norm{\cdot}}$, where $\mathbb{S}_{\norm{\cdot}} = \{\vc{w} \in \EE_0: \norm{ \vc{w}} = 1 \}$, and
\begin{equation}\label{eq:polar0}
T(\vc{X}) = \left( \norm{ \vc{X} }, \frac{\vc{X}}{ \norm{ \vc{X}}} \right) =: (R, \vc{W}).
\end{equation}
Then \eqref{eq:regvar} is equivalent to the existence of a finite measure $H_{\vX,\norm{\cdot}}$ (called the \emph{angular measure}) on $\mathbb{S}_{\norm{\cdot}}$ and a sequence $b_{n} \rightarrow \infty$ such that 
\begin{equation}\label{eq:polar}
n \, \PP \left[ \left( b_{n}^{-1} R, \vc{W} \right) \in \,  \cdot \, \right]  \vto \mu_{\alpha} \times H_{\vX,\norm{\cdot}}  (\, \cdot \,) , \qquad 
\text{as } n \rightarrow \infty,
\end{equation}
in the space of non-negative Radon measures on $(0,\infty] \times \mathbb{S}_{\norm{\cdot}}$, where $\mu_{\alpha}$ is a measure on $[0, \infty)$ given by
$\mu_{\alpha} \left( (x, \infty] \right) = x^{-\alpha}$ for $x > 0$ \citep[Theorem 6.1]{resnick2007}. {We will sometimes use the $L_{\alpha}$ norm, denoted $\norm{\, \cdot \,}_{\alpha}$, i.e., for $\vx \in [0, \infty)^d$, $\norm{\vx}_{\alpha} = \left( \sum_{j=1}^d x_j^{\alpha} \right)^{1/\alpha}$ for $\alpha \geq 1$. Whenever the choice of norm is obvious or does not matter, we will write $\mathbb{S}$ and $H_{\vX}$ without the subscript ${\norm{\, \cdot \,}}$.}

Let $B \subset \mathbb{S}$ be $H_{\vc{X}}$-measurable and consider the set 
$C_{s, B} := \{\vc{x} \in \EE_0: \norm{\vx}  > s,  \vx/\norm{\vx} \in B \}$.
Equation \eqref{eq:polar} implies $\nu_{\vX} (C_{s,B}) = s^{-\alpha} H_{\vX}(B)$ and $\nu_{\vX} (\diff r \times \diff \vc{w}) = \alpha r^{-\alpha - 1}\diff r \diff H_{\vX} (\vc{w})$. {Denote $m := \nu (C_{1,\mathbb{S}})= H_{\vc{X}} (\mathbb{S}) > 0$. Then immediately we have that 
\begin{equation} \label{eq:tail of R}
    \lim_{n\to\infty} n \, \PP[R > b_n x] = m x^{-\alpha }.
\end{equation}
 }
This relation implies that the choice of $b_n$ is not unique. Nevertheless, it must be a regularly varying function of $n$ with index $1/\alpha$. For simplicity, we choose $b_n=n^{1/\alpha}$ in the rest of the paper and assume that \eqref{eq:tail of R} holds. 
 Such a choice slightly limits the potential distributions of $R$ but is nevertheless in line with the max-linear models presented below.

{We remark that the angular measure $H_{\vc{X}}$ is related to the probability measure $S_{\vc{X}}$ in \citet[Chapter 6]{resnick2007} through $H_{\vc{X}} = m S_{\vc{X}}$} 
and that the marginals of $\vX$ satisfy
\begin{equation}\label{eq:scale}
\lim_{n \rightarrow \infty }n \, \mathbb{P} \left[  X_j > n^{1/\alpha} x \right] =    x^{- \alpha} \int_{\mathbb{S}} w_j^{\alpha} \diff H_{\vX} (\vc{w}).
\end{equation}

To summarize the dependence of a multivariate regularly varying vector $\vX$, we define the \emph{tail pairwise dependence matrix (TPDM)} $\Sigma_{\vc{X}}$ as 
\begin{equation}\label{eq:tpdm}
\Sigma_{\vX} = (\sigma_{\vX_{jk}})_{j,k=1,\ldots,\dm}, \qquad \text{with   } \,\,\, \sigma_{\vX_{jk}} = \int_{\mathbb{S} } w_j^{\alpha/2} w_k^{\alpha/2} \diff H_{\vX} (\vc{w}).
\end{equation}
We will write $\vc{X} \sim \Sigma_{\vX}$ if $\Sigma_{\vX}$ is the TPDM of $\vc{X}$.
When no confusion can arise, we will write the elements of $\Sigma_{\vc{X}}$ as $\sigma_{jk}$. Two variables $X_j,X_k$ are \emph{tail dependent} if and only if $\sigma_{jk} > 0$.  Recall that \emph{tail independence} corresponds to the scenario where 
all mass of the exponent measure $\nu_{\vX}$ is concentrated on the axes. 
From \eqref{eq:scale}, we observe that 
\begin{equation}\label{eq:scale2}
\lim_{n \rightarrow \infty }n \, \mathbb{P} \left[  X_j > n^{1/\alpha} x \right] = x^{-\alpha} \sigma_{{jj}}. 
\end{equation}
Hence, the (limiting) scale of $X_j$ is $\left( \sigma_{{jj}} \right)^{1/\alpha}$, and we say that the margins of $\vX$ are \emph{standardized} with index $\alpha$ whenever $\sigma_{jj} = 1$ for all $j=1,\ldots,d$.

Note that our definition of a TPDM is different from the one in \citet{larsson2012}, whose integrand does not depend on $\alpha$ (i.e. the two definitions coincide only when the margins of $\vX$ have been standardized to $\alpha = 2$); see also \citet{samorodnitsky1994}. The following proposition illustrates some \change{theoretical} properties of the TPDM. 
\begin{prop}\label{prop:properties}
Let $\Sigma_{\vX}$ be as defined in \eqref{eq:tpdm} and let $R$ and $\bm{W}$ be as defined in \eqref{eq:polar0}.
\begin{enumerate}
\item The elements of the TPDM can be written as 
$
\sigma_{\vX_{jk}} = m \lim_{x \rightarrow \infty} \EE \left[ W_j^{\alpha/2} W_k^{\alpha/2} \mid R > x \right].
$
\item The TPDM does not depend on the choice of norm; i.e., for any two norms $\norm{\, \cdot \,}$ and $\norm{\, \cdot \,}'$, we have 
\[
\sigma_{\vX_{jk}} =
\int_{\mathbb{S}_{\norm{\cdot}}} w_j^{\alpha/2}  w_k^{\alpha/2}  H_{\vX, \norm{\cdot}}  (\diff \vc{w}) = \int_{\mathbb{S}_{\norm{\cdot}'}}  (w_j')^{\alpha/2} (w_k')^{\alpha/2}  H_{\vX, \norm{\cdot}'}  (\diff \vc{w}').
\]
\item Let $\vX$ be multivariate regularly varying with tail index $\alpha$ and angular measure $H_{\vX}$. If 
$\vX^* = (X^*_1,\ldots,X^*_d)^T = (X_1^{\alpha/\alpha^*},\ldots,X_d^{\alpha/\alpha^*})^T$,
then $\vX^*$ is multivariate regularly varying with tail index $\alpha^* > 0$. In addition, if $H_{\vX^*}$ is the angular measure of $\vX^*$ obtained using $b_n^*=n^{1/\alpha^*}$, we have  
$\Sigma_{\vX} = \Sigma_{\vX^*}$. 
\item Let $\vX$ and $\vY$ be two independent multivariate regularly varying vectors with the same tail index $\alpha$. The corresponding angular measures are  $H_{\vX}$ and $H_{\vY}$ respectively, when choosing $b_n=n^{1/\alpha}$. Then the component-wise maximum $\max(\vX, \vY)$ is multivariate regularly varying with the same tail index $\alpha$ and with angular measure $H_{\vX}+H_{\vY}$. Moreover, the TPDM of $\max(\vX, \vY)$ is 
$\Sigma_{\max(\vX, \vY)}=\Sigma_{\vX}+\Sigma_{\vY}$.
\end{enumerate}
\end{prop}
Proposition~\ref{prop:properties} shows that the TPDM is invariant both to the choice of norm and to the choice of the (marginal) tail index. Using the $L_{\alpha}$ norm, we see that the total mass of the angular measure equals
$m = H_{\vX}(\mathbb{S}) = \int_{\mathbb{S} }\norm{\vc{w}}_{\alpha}^{\alpha} \diff H_{\vX} (\vc{w})  = \sum_{j=1}^{\dm}  \sigma_{{jj}}$.


\subsection{Max-linear model}\label{sec:maxlin}
We are interested in calculating the probability that $\vX$ falls into some extreme failure region  $C \subset \EE_0$. In practice, to estimate such probabilities $\PP[\vX \in C]$, we need to assume a parametric model for the exponent measure $\nu$, or, equivalently, for the angular measure $H$. A simple yet effective parametric model for the goal of estimating failure probabilities is the \emph{max-linear model}. 

Let $A$ denote a $\dm \times \fac$ matrix with non-negative entries, with columns $\vc{a}_{1},\ldots,\vc{a}_{\fac}$, and suppose that $\max(\vc{a}_l) = \max(a_{1l}, \ldots, a_{dl}) > 0$ for all $l \in \{1,\ldots,\fac\}$. Let $Z_1,\ldots,Z_q$ be independent random variables whose distribution is Fr\'echet with scale 1 and shape $\alpha$, \change{i.e., $Z_1, \ldots, Z_q \siid \textnormal{Fréchet}(\alpha)$ and $\PP[Z_l \leq z] = \exp \left\{ -z^{-\alpha}\right\}$,} and define
$$\vc{Y} = A \times_{\max} \vc{Z}  := \left( \max_{l=1,\ldots,q} a_{1l} Z_l, \ldots,  \max_{l=1,\ldots,q} a_{dl} Z_l \right)^T,$$
for $\vc{Z} = (Z_1,\ldots,Z_q)$. Then $\vc{Y}$ is multivariate regularly varying with an index $\alpha$ and the angular measure of $\vc{Y}$ is
\begin{equation}\label{eq:maxlinH}
 H_{\vc{Y}} (\, \cdot \,) = \sum_{l=1}^{\fac} \norm{\vc{a}_{l}}_{\alpha}^{\alpha} \, \delta_{\vc{a}_{l}/\norm{\vc{a}_{l}}_{\alpha}} ( \, \cdot \,).
\end{equation}
Note that the order of the columns of $A$ does not matter.
The marginal distribution of $Y_j$ is Fr\'echet with scale $\left( \sum_{l=1}^q a_{jl}^{\alpha} \right)^{1/\alpha}$ and shape $\alpha$ for $j = 1,\ldots,d$.  The TPDM of $\vc{Y}$ has elements 
\[
\sigma_{\vc{Y}_{jk}} = \sum_{l=1}^q \norm{\vc{a}_l}_{\alpha}^{\alpha} \left( \frac{a_{jl}}{\norm{\vc{a}_l}_{\alpha }}\right)^{\alpha/2}  \left( \frac{a_{kl}}{\norm{\vc{a}_l}_{\alpha }}\right)^{\alpha/2}  = \sum_{l=1}^q a_{jl}^{\alpha/2} a_{kl}^{\alpha/2},
\]
and hence $\Sigma_{\vY} =  A_* A^T_*$, where  $A_* := \left( a_{jl}^{\alpha/2} \right)_{j=1,\ldots,d, l = 1,\ldots,q}$.

\citet{cooley2019} show that the class of max-linear angular measures, as defined in \eqref{eq:maxlinH}, is dense in the class of possible angular measures (for any norm),  i.e., any multivariate regularly varying vector $\vX$ can be approximated by a max-linear model with a $(d \times q)$ parameter matrix as $q \rightarrow \infty$. 
Moreover, if attention is restricted to the TPDM, a max-linear model with \emph{finite} $q$ is sufficient to exactly match $\Sigma_{\vX}$; this result has been proven in \citet{cooley2019} for $\alpha = 2$, but can be easily generalized to any $\alpha>0$.
\begin{prop}[\citet{cooley2019}, Propositions 4 \& 5]
\label{prop:maxlin}
\text{}
\begin{enumerate}
\item \change{Let $Z_1, \ldots, Z_q \siid \textnormal{Fréchet}(\alpha)$}. Given any angular measure $H$, there exists a sequence of non-negative $(d \times q)$ matrices $\{ A_q \}$, for $q = 1, 2, \ldots$ such that $H_{A_q \times_{\max} \vc{Z}} \rightarrow H$ weakly as $q \rightarrow \infty$.
\item If $\vX \in \RR^d$ has TPDM $\Sigma_{\vX}$, there exists a non-negative $(d \times \widetilde{q})$ matrix $A_*$ with $\widetilde{q} < \infty$, such that $\Sigma_{\vX} = A_* A_*^T$. 
\end{enumerate}
\end{prop}

Whenever a matrix $\Sigma$ can be decomposed as $\Sigma=AA^T$, and the matrix $A$ has non-negative entries, we say that $\Sigma$ is \emph{completely positive}. {The max-linear model construction implies that any completely positive matrix can be a TPDM. Conversely, the second part of Proposition \ref{prop:maxlin} shows that any TPDM is completely positive. In fact, for any vector $\vX$, there exists a max-linear vector $\vY$ whose TPDM is equal to that of $\vX$.} This result is key for the estimation of failure probabilities.
{
\begin{rem}
The class of completely positive matrices coincides with the class of TPDMs. In particular, this implies that a TPDM is always positive semi-definite and element-wise non-negative \citep[Chapter 3]{shaked2021copositive}. Note that \cite{embrechts2016bernoulli} characterized the class of matrices of \emph{tail dependence coefficients} $(\chi_{jk})_{j,k=1,\ldots,d}$, defined by
$\chi_{jk} = \lim_{t \downarrow 0} t^{-1} \PP[F_j(X_j) > 1-t, F_k(X_j) > 1-t]$, 
and found that it is a strict subset of the class of completely positive matrices. \change{A related problem focuses on the sub-class of max-stable processes that is characterized through \emph{extremal coefficients} 
\citep{strokorb2015,fiebig2017,shyamalkumar2020tail,janssen2023tail}}. 
\end{rem}}

\subsection{The probability of a failure region under the max-linear model}\label{sec:failure}
\change{Given a multivariate regularly varying random vector $\vX$, via a suitable completely positive decomposition of its TPDM, one can construct a max-linear vector $\vY$, parameterized by the matrix $A$, whose TPDM is equal to that of $\vX$. Then one can approximate the probability of the failure region $\PP[\vc{X} \in C]$ by the failure probability $\PP[\vc{Y} \in C]$}. Indeed, whenever the set $C \subset \EE_0$ is ``extreme enough'', we have  
\begin{equation}\label{eq:failure}
\mathbb{P} \left[ \vc{Y} \in C \right] = \frac{1}{n} \left\{ n \,  \mathbb{P} \left[ n^{-1/\alpha} \vc{Y} \in n^{-1/\alpha} C \right]  \right\} \approx \frac{1}{n} \nu_{\vY} \left( n^{-1/\alpha} C \right).
\end{equation}
\change{Here the failure set $C$ is scaled back to $n^{-1/\alpha} C$, which is no longer an extreme region, and the corresponding exponent measure $\nu_{\vY}(n^{-1/\alpha} C)$ can be derived from the parameters of the max-linear model.} 

The choice of the region $C$ in \eqref{eq:failure} depends on the application of
interest. For some practically relevant failure regions, the \change{exponent measure can be calculated analytically}. 
For environmental data, e.g., if $\vc{Y}$ represents a vector of maximum daily wind speeds, rainfall amounts, or temperatures in $d$ locations, the region 
$C_{(\max)} (\vx)  = \{ \vy \in \EE_0: y_1 > x_1 \text{ or } \ldots \text{ or } y_d > x_d \}$,
is a common choice if an extreme event at at least one location can lead to a climatological catastrophe.
The exponent measure for this region equals
\begin{align*}
\nu( C_{(\max)} (\vx) ) = \sum_{l=1}^q \max_{j=1,\ldots,d} \left( \frac{a_{jl}}{x_j} \right)^{\alpha}.
\end{align*} 
\change{A second example is in portfolio risk management. If} $\vc{Y}$ represents for example the negative log-returns of $d$ stocks in a portfolio, one might be interested in estimating the probability of all stocks crashing together, 
$ C_{(\min)} (\vx)  = \{ \vy \in \EE_0: y_1 > x_1 \text{ and } \ldots \text{ and } y_d > x_d \}$.
Similarly to $C_{(\max)}$, we can calculate
\begin{align*}
\nu( C_{(\min)} (\vx) ) = \sum_{l=1}^q \min_{j=1,\ldots,d} \left( \frac{a_{jl}}{x_j} \right)^{\alpha}.
\end{align*} 

\change{More generally, m}any simple regions (including the ones above when $x_1 = \ldots = x_d$) can be written as follows. Let $C_f (x) = \{ \vc{y} \in \EE_0: f(\vc{y}) > x\}$ denote a failure region where $f: \RR^d \rightarrow \RR$ is {a measurable function and is non-decreasing in all dimensions}. For fixed $\bm{w}$, let $r_* = r_*(\vc{w})$ be such that $f(r_* \vc{w}) = x$. Then
\begin{equation*}
\nu (C_f (x)) = \sum_{l=1}^q \norm{\vc{a}_l}^{\alpha}_{\alpha} \left\{ r_*( \vc{a}_l / \norm{\vc{a}_l}_{\alpha}) \right\}^{-\alpha}.
\end{equation*}
Hence, we can obtain an explicit expression for the exponent measure of any failure region that can be written in the form of $C_f$.  For example, if $f(\vy)= \vc{v}^T \vc{y}$ for $\vc{v} \in \EE_0$ such that $v_1 + \ldots + v_d = 1$, we get 
$ C_{(\summ)} (\vc{v},x)  = \{ \vc{y} \in \EE_0: v_1y_1 + \ldots + v_d y_d > x \}$.
\change{Such a region is relevant for climate applications, where $\sum_j v_j X_j$ represents a spatially aggregated
environmental variable \citep{kiriliouk2020nr2,engelke2019nr2}, and in a financial context, if one intends to estimate the Value-at-Risk or Expected Shortfall of a portfolio loss $L = \sum_{j=1}^d v_j X_j$, where $v_1,\ldots,v_d$ represent weights of the stocks \citep[see, e.g.,][]{mainik2013}. We have that 
\begin{align*}
\nu \left( C_{(\summ)} (\vc{v}, x)\right) =  x^{- \alpha} \sum_{l=1}^q \left( \vc{v}^T \vc{a}_l  \right)^{\alpha} 
=
\begin{dcases} 
  x^{-1}  \sum_{j=1}^d v_j \sigma_{jj} & \text{if } \alpha = 1, \\
  x^{-2} \left( \vc{v}^T \Sigma_{\vc{Y}} \vc{v} \right) & \text{if } \alpha = 2.
 \end{dcases}
\end{align*} 
Hence, for $\alpha \in \{1,2\}$, we can calculate} $\nu \left( C_{(\summ)} (\vc{v}, x) \right)$ based on the TPDM only. In other words, the random variable $\vc{v}^T \vY$ approximately follows a Pareto distribution whose scale is a function of the TPDM. By contrast, when $\alpha \notin \{1,2\}$, the parameter matrix $A$ needs to be recovered first. 

\change{For more complex failure regions $C$ we may be unable to derive an analytic expression for $\nu(C)$, but we can still calculate the exponent measure by simulating observations from the associated max-linear model. For such failure regions, recovering the parameter matrix $A$ is essential for estimating $\PP(\vX \in C)$.}

\subsection{Multivariate domain of attraction}\label{sec:general}

\change{So far we have assumed that the random vector $\vX$ is multivariate regularly varying with index $\alpha$. This assumption is equivalent to $\vX$ belonging to the \emph{(max-)domain of attraction} of a multivariate generalized extreme value (GEV) distribution $G$ with Fr\'echet$(\alpha)$ margins $G_j(x) = \exp \{ - x^{-\alpha} \}$ for $j = 1,\ldots,d$ \citep[Proposition 5.17]{resnick1987}. In other words, there exist normalizing sequences $\vc{b}_n \in \RR^d$, $\vc{a}_n \in [0,\infty)^d$ such that
$F^n \left( \vc{a}_n \vx + \vc{b}_n \right) \rightarrow G(\vx)$ for all $\vx$ in the support of $F$, in which case we write $F \in \textnormal{MDA}(G)$. It follows that 
$X_1,\ldots,X_d$ are univariate regularly varying with common tail index $\alpha$ for all $j = 1,\ldots,d$.



In practice, the assumption of multivariate regular variation assumption may be restrictive; marginal distributions of $\vX$ may have unequal tail indices or may not even be regularly varying. A weaker assumption, under which our method of estimating failure probabilities is still applicable, is the general \emph{multivariate domain of attraction condition}, i.e., $F \in \textnormal{MDA}(\tilde{G})$ where $\tilde{G}$ is a multivariate GEV distribution with arbitrary margins $\tilde{G}_1,\ldots,\tilde{G}_d$. 
Setting  
\begin{align}\label{transformation}
  X_{j}^* = \left(- \frac{1}{\log F_j(X_{j})} \right)^{1/\alpha}, \quad j = 1,\dots, d,
\end{align}
we obtain a vector $\vX^* = (X_1^*,\ldots,X_d^*)^T$ with Fréchet$(\alpha)$ margins. By Proposition 5.10 in \cite{resnick1987}, such a marginal standardization does not affect the dependence structure, and $F \in \textnormal{MDA}(\tilde{G})$ implies that 
$\vX^*$ is multivariate regularly varying with index $\alpha$. 
Hence, if the assumption of multivariate regular variation is not satisfied for a given vector $\vX$, our methodology can be applied to the standardized vector $\vX^*$. 

In practice, one needs to estimate the marginal distributions $F_1,\ldots,F_d$ from the data first; see Section~\ref{sec:estim}. Then one can construct a max-linear model $\vY^*$ whose TPDM is equal to that of $\vX^*$. Recall that our actual goal is to estimate $\PP[\vX \in C]$. 
For some failure regions, such as $C_{(\max)}$ or $C_{(\min)}$, it is straightforward to find $C^*$ such that $\PP[\vX \in C] = \PP[\vX^* \in C^*]$ since transformation~\eqref{transformation} can be applied component-wise.  We can then proceed as in Section~\ref{sec:failure}, approximating $\PP[\vX^* \in C^*]$ by $\PP[\vY^* \in C^*]$. Other regions such as $C_{(\textnormal{sum})}$ are untractable after marginal transformation, but we can still apply the inverse of \eqref{transformation} to $\vY^*$ to obtain a random vector $\vY$ that approximates the original data vector $\vX$.
Based on a large number of simulated observations from $\vY$, we can then estimate $\PP[\vX \in C]$ by the empirical frequency of the simulated observations falling into $C$. This approach has been adopted in various applications; see, e.g., \citet{cai2013environmental}. We use it in Section \ref{sec:wind}.
Of course, if multivariate regular variation is a reasonable assumption for the random vector $\vX$, it is preferred to use the direct approach without standardizing a priori.}




\section{Decomposing the pairwise tail dependence matrix}\label{sec:alg}

Proposition \ref{prop:maxlin} shows that, for any random vector $\vX $, a max-linear vector $\vc{Y}$ exists such that $\Sigma_{\vX} = \Sigma_{\vY}$. Obtaining such a model requires a \emph{completely positive decomposition} of $\Sigma_{\vX}$, i.e., finding a $(d \times q)$ matrix $A_*$ with non-negative entries such that $\Sigma_{\vX}   = A_* A_*^T$ (see also Section \ref{sec:maxlin}). 
In this section, we propose a simple algorithm to construct a $(d \times d)$ matrix $A_*$, which leads to an \emph{approximate} completely positive decomposition in the following sense. Firstly, we define the notation ``$\preceq$'': for two TPDMs, we write $\Sigma_{\vX} \preceq \Sigma_{\vY}$ if
\[
\sigma_{\vc{X}_{jk}} = \sigma_{\vc{Y}_{jk}} \,\,\,\,\, \text{for all } j,k \in \{1, \ldots, d \}, \, j \neq k, \quad \text{and } \, \sigma_{\vc{X}_{jj}} \leq \sigma_{\vc{Y}_{jj}} \,\,\,\,\,   \text{for all } j \in \{1, \ldots, d \}.
\]
Then, for any random vector $\vX$ with TPDM $\Sigma_{\vX}$, we aim at constructing a max-linear vector $\vc{Y}$ with a $(d \times d)$ coefficient matrix $A_*$, such that $\Sigma_{\vX} \preceq \Sigma_{\vY}$. The max-linear model can be viewed as an approximate model for the original random vector $\vX$, ``matching'' the pairwise dependence structure of $\vX$. 

\subsection{The algorithm}\label{sec:algorithm}
Define, for any TPDM $\Sigma_{\vX} = (\sigma_{{jk}})_{j,k=1,\ldots,\dm}$ and any $i \in \{1,\ldots,\dm\}$,
\begin{equation*}
D_i (\Sigma_{\vX}) :=   \max \left\{j,k \in \{1,\ldots,\dm\} \setminus \{i\}: \, \frac{\sigma_{ji} \sigma_{ki}}{\sigma_{jk} \sigma_{ii}} \right\}.
\end{equation*}
We write $D_i = D_i (\Sigma_{\vX})$ when no confusion can arise. Let $\vc{\tau}_i = (\tau_{1,i},\ldots,\tau_{\dm,i})^T$ with
\begin{equation} \label{eq:tau}
\tau_{j,i}  = 
\begin{dcases}
\sigma_{ji} \left(\sigma_{ii} \max(D_i,1)\right)^{-1/2} & \text{ if } j \neq i,\\
\left(\sigma_{ii} \max(D_i,1)\right)^{1/2} & \text{ if } j = i.
\end{dcases}
\end{equation}
Let $\vc{\tau}_{-i} := (\tau_{1,i},\ldots,\tau_{i-1,i},\tau_{i+1,i},\ldots,\tau_{\dm,i})^T  \in \RR^{\dm-1}$ and define
\begin{equation}\label{eq:sigmai}
\Sigma^{(i)}_{\vX} := \Sigma^{(-i,-i)}_{\vX} - \vc{\tau}_{-i} \vc{\tau}_{-i}^T \quad \in \RR^{(\dm-1) \times (\dm - 1)},
\end{equation}
where $\Sigma^{(-i,-i)}_{\vX} \in \RR^{(\dm-1) \times (\dm - 1)}$ is the matrix obtained by removing the $i$-th row and column from $\Sigma_{\vX}$. 
\begin{lem}\label{lem1}
For all $i \in \{1,\ldots,\dm\}$ with $D_i \neq 1$, the matrix $\Sigma^{(i)}_{\vX}$ in \eqref{eq:sigmai} {satisfies that}
\begin{enumerate}
\item $\Sigma^{(i)}_{\vX} \geq 0$ componentwise;
\item $\Sigma^{(i)}_{\vX}$ is positive semi-definite.
\end{enumerate}
\end{lem}
The next proposition shows that if we can construct a $(d-1)$-dimensional random vector with TPDM ${\tilde\Sigma^{(i)}_{\vY}\succeq} \Sigma^{(i)}_{\vX}$, we can then use this vector to construct $\vY$ such that \change{$\Sigma_{\vY} \succeq \Sigma_{\vX}$}. In addition, the proposition illustrates the circumstances under which \change{$\Sigma_{\vY}=\Sigma_{\vX}$} is achieved. 

\begin{prop}\label{prop:peel}
Let $\Sigma_{\vX} = (\sigma_{jk})_{j,k=1,\ldots,d} \in \RR^{d \times d}$ be a TPDM. {Let $i \in \{1,\ldots,d\}$. Consider a random vector $\vc{\tilde Y}_{-i} = (\tilde Y_1,\ldots,\tilde Y_{i-1},\tilde Y_{i+1},\ldots,\tilde Y_{\dm})^T$ with TPDM $\tilde\Sigma^{(i)}_{\vY}\succeq \Sigma^{(i)}_{\vX}$, where $\Sigma^{(i)}_{\vX}$ is defined as in \eqref{eq:sigmai}. \change{Let $Z_i \sim \textnormal{Fréchet}(\alpha)$ be independent of} $\tilde \vY_{-i}$. Defining a new vector $\vY$ via}
\begin{equation*}
Y_j  := \begin{dcases}
\max \left( \tau_{j,i}^{2/\alpha} Z_i , \tilde Y_j \right) & \text{ if } j \in \{1,\ldots,\dm\} \setminus \{i\}, \\
\tau_{i,i}^{2/\alpha} Z_i & \text{ if } j = i,
\end{dcases}
\end{equation*}
{with $\tau_{j,i}$ as in \eqref{eq:tau},
we have that \change{$\Sigma_{\vY} \succeq \Sigma_{\vX}$}. Moreover,  $\Sigma_{\vY} = \Sigma_{\vc{X}}$ if and only if $D_i (\Sigma_{\vX}) \leq 1$  and $\tilde\Sigma^{(i)}_{\vY}=\tilde\Sigma^{(i)}_{\vX}$.}
\end{prop}
{Based on Lemma \ref{lem1} and Proposition \ref{prop:peel},} we can repeatedly apply the procedure in expression \eqref{eq:sigmai} to $\Sigma^{(i)}_{\vX}$ to further decrease its dimension. This leads to an iterative algorithm, throughout which we obtain the matrix $A_*$. First, let $i_1 \mapsto i_2 \mapsto \ldots \mapsto i_d$ denote a \emph{path}, where $(i_1,\ldots,i_d)$ is a permutation of $(1,\ldots,d)$. The order determines which column of $\Sigma_{\vX}$ will be treated first. For a given path, the iterative algorithm is as follows.
\begin{itemize}
\item In the first step, we obtain the vector $\vc{\tau}_{i_1}$ from \eqref{eq:tau} by taking $i=i_1$. We fill in the first column of the matrix $A_*$ with the obtained vector $\vc{\tau}_{i_1}$ (as the order of the columns of $A_*$  does not matter, we do not necessarily need to fill the $i_1$-th column). The targeted TPDM is then reduced to a $(d-1)\times (d-1)$ matrix {$\Sigma_{\vc{X}}^{(i_1)}$} as defined in \eqref{eq:sigmai}.
\item In the second step, we apply expressions \eqref{eq:tau} and \eqref{eq:sigmai} to {$\Sigma_{\vc{X}}^{(i_1)}$}, with the focal dimension now indexed by $i_2$. This step results in a $(d-1)$-dimensional vector. For the second column of the matrix $A_*$, we set the $i_1$-th row to zero and fill in the rest with the obtained $(d-1)$-dimensional vector. After the second step, the targeted TPDM is reduced to a $(d-2)\times(d-2)$ matrix.
\item The previous step is iterated as follows. In each step $j$, \change{we aim to fill the $j$-th column of $A_*$}. First, we set the elements in the $i_1,i_2,\cdots,i_{j-1}$-th row to zero. Then, we fill in the other $(d-j+1)$ elements by the obtained $(d-j+1)$-dimensional vector of this iteration. After each step, the dimension of the targeted TPDM is reduced by one.
\item At step $d$, the targeted TPDM \change{is reduced to a single} positive number. \change{We simply fill element $(i_d,d)$ of $A_*$ with the square root of that value and set all other elements in the $d$-th column to zero.}
\end{itemize}
Proposition \ref{prop:peel} guarantees that the matrix $A_*$ obtained with this algorithm has only non-negative elements and satisfies $\Sigma_{\vX}\preceq A_*A_*^T$. In addition, \change{define $A=A_*^{2/\alpha}$ element-wise
and let $\vc{Z}=(Z_1,\ldots,Z_d)^T$ where 
$Z_1, \ldots, Z_d \siid \textnormal{Fréchet}(\alpha)$.
The max-linear vector $\vY=A \times_{\max} \vc{Z}$ has a TPDM $\Sigma_{\vY} = A_* A_*^T$ that satisfies $\Sigma_{\vX}\preceq \Sigma_{\vY}$.} In other words, we can construct a max-linear vector $\vY$ based on only $d$ independent factors to approximate the TPDM of $\vX${, with at least all off-diagonal elements exactly matched.}

Below is an example for this procedure when $d=3$ and the chosen path is $1 \mapsto 2 \mapsto 3$.
\begin{exa}[$d = 3, \alpha = 2$, path $1 \mapsto 2 \mapsto 3$]\label{example:3d}
Let $\Sigma_{\vX} = (\sigma_{jk})_{j,k = 1,2,3}$ be a known (or estimated) TPDM. We want to find a coefficient matrix $A$ (since $\alpha=2$, we have $A=A_*$) such that the corresponding max-linear model $\vY =  A \times_{\max} (Z_1,Z_2,Z_3)^T$ has TPDM $\Sigma_{\vY}$ satisfying $\Sigma_{\vX} \preceq \Sigma_{\vY}$. Starting with $i  =1 $, we obtain from \eqref{eq:tau} that
\[
\begin{pmatrix} \tau_{1,1} \\ \tau_{2,1} \\ \tau_{3,1} \end{pmatrix} = \frac{1}{\sqrt{\sigma_{11} \max(D_1,1)}} \begin{pmatrix} \sigma_{11} \max(D_1,1) \\ \sigma_{12} \\ \sigma_{13} \end{pmatrix}
\]
and 
\begin{equation*}
\Sigma_{\vX}^{(-1)} := \begin{pmatrix}
s_{22} & s_{23} \\
s_{23} & s_{33} \\
\end{pmatrix}
:=
\begin{pmatrix}
\sigma_{22} - \frac{\sigma_{12}^2}{\sigma_{11} \max(D_1,1)} & \sigma_{23} - \frac{\sigma_{12} \sigma_{13}}{\sigma_{11} \max(D_1,1)} \\
\sigma_{23} - \frac{\sigma_{12} \sigma_{13}}{\sigma_{11} \max(D_1,1)} & \sigma_{33} - \frac{\sigma_{13}^2}{\sigma_{11} \max(D_1,1)} \\
\end{pmatrix}.
\end{equation*}
Hence, the first column of $A$ is given by $(\tau_{1,1}, \tau_{2,1}, \tau_{3,1})^T$.

To find the second column of $A$, we iterate the same algorithm to the newly obtained TPDM $\Sigma_{\vX}^{(-1)}$ to get that 
\[
\begin{pmatrix} \tau_{2,2} \\ \tau_{3,2} \end{pmatrix} = \frac{1}{\sqrt{s_{22} }} \begin{pmatrix} s_{22}  \\ s_{23} \end{pmatrix}.
\]
Note that since $\Sigma_{\vX}^{(-1)}$ is positive semi-definite, in this case we have $D_2(\Sigma_{\vX}^{(-1)})=\frac{s_{23}^2}{s_{33}s_{22}}\leq 1$. After this iteration, the targeted TPDM is reduced to a single value $s_{33} - s_{23}^2/ s_{22}\geq 0$. Hence, we can define $\tau_{3,3}:= \sqrt{s_{33} - s_{23}^2/ s_{22}}$. 

\change{Finally,} we can construct the max-linear model $\vc{Y} = A \times_{\max} (Z_1,Z_2,Z_3)$ where
\begin{equation*}
A = \begin{pmatrix}
\tau_{1,1} & 0 & 0 \\
\tau_{2,1} & \tau_{2,2} & 0 \\
\tau_{3,1} & \tau_{3,2} & \tau_{3,3} \\
\end{pmatrix}, \qquad \text{and} \qquad \Sigma_{\vc{Y}} = A A^T = 
 \begin{pmatrix}
\sigma_{11} \max(D_1,1) & \sigma_{12} & \sigma_{13} \\
\sigma_{12} & \sigma_{22} & \sigma_{23} \\
\sigma_{13} & \sigma_{23} & \sigma_{33}
\end{pmatrix}.
\end{equation*}
\end{exa}

In Example \ref{example:3d}, we applied the iterative algorithm implied by Proposition \ref{prop:peel} with the path $1 \mapsto 2 \mapsto 3$. The resulting matrix $A_*$ (and eventually $A$) is then lower triangular. By considering other paths we will obtain different matrices $A_*$. Hence, in general, for different choices of paths, a completely positive decomposition (approximate or not) is not unique. It is also possible that one path leads to an exact decomposition while another path leads to an approximation in the sense of $\preceq$. Proposition \ref{prop:unique} shows that if the vector $\vX$ indeed follows a max-linear model constructed from a lower-triangular parameter matrix $A$, by choosing the path \emph{appropriately}, it is possible to recover the exact max-linear model. Given that the order of columns of $A$ does not matter, the proposition can of course be generalized to parameter matrices $A$ in which each column $j$ has at least one more zero than column $(j-1)$, for $j=2,\ldots,d$.
\begin{prop}\label{prop:unique}
Suppose that $\vc{X} = A \times_{\max} \vc{Z}$ follows a max-linear model with lower-triangular parameter matrix $A \in \RR^{d \times d}$ and \change{$\vc{Z} = (Z_1,\ldots,Z_q)^T$, where $Z_1, \ldots, Z_q \siid \textnormal{Fréchet}(\alpha)$}. If we choose the canonical path, i.e., $1 \mapsto \ldots \mapsto d$ in the decomposition of $\Sigma_{\vc{X}}$, {then we obtain the matrix $A_* = A^{\alpha/2}$, with $A_*{A_*}^T=\Sigma_{\vc{X}}$.}
\end{prop} 
\begin{rem} \label{rem:dag}
\change{When the matrix $A$ is lower-triangular, the max-linear model $\vc{X} = A \times_{\max} \vc{Z}$ corresponds to a fully connected directed acyclic graph (DAG), where each node $i$ points to all nodes $j>i$; see e.g. \cite{krali2023heavy}. In other words, with a known DAG structure, our algorithm provides the edge weights of a given TPDM. When all edge weights are non-zero, the DAG is a complete graph. Each zero element in the lower-triangular part of $A$ indicates the absence of an edge in the DAG.
When the graphical structure is unknown, if the proposed algorithm leads to an exact decomposition, 
it reveals one possible DAG structure (i.e. a potential causal structure) for the components of $\vX$ that matches the TPDM of $\vX$. } 
\end{rem}
\begin{rem}\label{rem:zerosigma}
{
   If the matrix $\Sigma_{\vX}$ has a zero eigenvalue, it means that the components of $\vX$ are linearly related when the radius $R$ is high. In practice, an eigenvalue of $\Sigma_{\vX}$ can indeed be close to zero when the variables are strongly tail dependent. This is comparable with the concept of multicolinearity in multivariate statistics. In that case, we will encounter a zero element in $\Sigma_{\vX}^{(i)}$ after a certain number of steps (say, $d'$) of the iterative algorithm. As a consequence, the matrix $A_*$ will have $d'$ columns, where $d'$ equals the rank of $\Sigma_{\vX}$.}  
\end{rem}

We conclude this subsection by discussing the usefulness of a non-exact decompositions in downstream tasks such as estimating the probability of a failure set. Consider two TPDMs, $\Sigma_{\vY}\succeq \Sigma_{\vX}$, with $\Sigma_{\vY} \neq \Sigma_{\vX}$. Denote $\vc{w}=(\text{diag}(\Sigma_{\vY})-\text{diag}(\Sigma_{\vX}))^{1/\alpha}$, where the $\text{diag}$ operator is defined by taking the diagonal elements of a square matrix. Then $\vc{w}$ is a non-negative $d-$dimensional vector with at least one positive element. \change{Suppose that $\vX$ is multivariate regularly varying with index $\alpha$ and TPDM $\Sigma_{\vX}$. Let $Z_1, \ldots, Z_d \siid \textnormal{Fréchet}(\alpha)$ 
and define $\vY= \left(\max(X_1, w_1 Z_1), \ldots, \max(X_d, w_d Z_d) \right)^T$. From Proposition \ref{prop:properties}, Statement 4, we obtain $\Sigma_{\vY} = \Sigma_{\vX}+\text{Diag}(\vc{w}^{\alpha})$, where the $\text{Diag}$ operator constructs a diagonal matrix using the components of $\vc{w}^{\alpha}$. In other words, there exists a random vector $\vY$ with TPDM $\Sigma_{\vY}$ such that $\vY\geq \vX$ element-wise. Consider a failure set $C_f(x)=\{\vx \in \RR^d : f(\vx)> x\}$ where $f$ is measurable and non-decreasing in all dimensions. We immediately get $f(\vY)\geq f(\vX)$, which implies $\PP[\vY \in C_f(x) ]\geq \PP[\vX \in C_f(x) ]$, i.e., using the constructed vector $\vY$, we can only overestimate the intended failure probability. Hence, for regions of the form $C_f(x)$, using a non-exact decomposition can only lead to a conservative failure probability estimate, which can be desirable from the perspective of risk management.}

\subsection{Applying the algorithm to estimate failure probabilities}\label{sec:practical}
\change{The decomposition algorithm in Section \ref{sec:algorithm} requires a specification of the path a priori. With a suitable choice of the path, one may obtain a decomposition leading to a TPDM that is close to, or even exactly equal to, the original TPDM.} When handling a TPDM $\Sigma_{\vX}$ with a large $d$, it is not feasible to exhaust all $d!$ possible paths in search for an exact decomposition. 
In practice, we need an efficient searching algorithm to find \change{such} a path, if possible. We consider three potential algorithms to achieve this goal. 
\begin{enumerate}
\item A simple approach: in step $j$ we pick the next dimension $i_j$ corresponding to the lowest value of $D_i$ among all remaining dimensions. 
\item An exhaustive approach: in step $j$, we let $T_j$ denote the \emph{set} of remaining dimensions corresponding to a value of $D_i$ below 1, i.e. $T_j := \{ i : D_i \leq 1, i\neq i_1,i_2,\cdots,i_{j-1}\}$. For each $i\in T_j$, we proceed with the iterative process for $\Sigma^{(-i)}_{\vX}$. This will result in a tree of possibilities after each step. For some ``branches'' of the tree, the algorithm may end at a ``leaf'' where no $D_i$ lies below 1, i.e. $T_j=\varnothing$. If some ``branch'' can iterate $d$ steps, it leads to a path that corresponds to an exact decomposition. We retain all paths (and hence all possible matrices $A_*$) that lead to an exact decomposition. 
\item A pragmatic approach: in step $j$, we \change{let $T_j := \{ i : D_i < 1, i\neq i_1,i_2,\cdots,i_{j-1}\}$ denote the \emph{set} of remaining dimensions corresponding to a value of $D_i$ below 1}. We randomly choose an element from this set and use it to proceed with the iteration. If the procedure stops after less than $d$ iterations and ends with a scenario where no $D_i$ lies below 1, we restart the entire procedure from the first iteration. We repeat the algorithm until finding an exact decomposition. 
\end{enumerate}
Note that the simple approach may not lead to an exact decomposition because there is no guarantee that the minimum of the $D_i$'s is below $1$ in all steps. The exhaustive approach will find all exact decompositions if there exists at least one. However, it is limited to moderate dimensions $d$ because of its computational intensity. \change{The pragmatic approach does not guarantee finding an exact decomposition either, but for the real data examples in Section \ref{sec:appl}, an exact decomposition is quickly found after a few ``restarts'', and hence it is the approach that we use in practice.}

\change{To investigate the accuracy of our algorithm in approximating failure probabilities, we design a numerical experiment} as follows, inspired by Section 2 of the supplementary material of \cite{cooley2019}. {Define
\footnotesize
\begin{equation*}
A_1 = \begin{pmatrix}
1.00 & 0.50 & 0.00 & 0.25 & 1.75 & 0.50 & 0.75 & 1.00 & 1.00 & 0.25 & 1.75 & 0.25 & 1.50 & 0.50 & 0.25 & 1.50 \\ 
2.00 & 0.00 & 1.50 & 1.00 & 1.00 & 0.25 & 1.00 & 1.00 & 1.75 & 0.25 & 0.00 & 2.00 & 0.25 & 0.50 & 0.25 & 0.75 \\ 
1.75 & 1.25 & 0.75 & 0.25 & 0.50 & 2.00 & 1.75 & 0.25 & 0.75 & 0.25 & 1.25 & 0.25 & 1.00 & 0.75 & 1.00 & 0.25 \\ 
1.25 & 0.25 & 2.00 & 0.25 & 1.25 & 2.00 & 0.50 & 0.25 & 0.50 & 0.50 & 0.00 & 0.75 & 0.50 & 0.25 & 1.75 & 0.50 \\ 
1.75 & 0.50 & 0.75 & 1.25 & 0.25 & 0.50 & 1.75 & 0.00 & 2.00 & 1.00 & 1.50 & 0.50 & 0.00 & 0.50 & 1.25 & 1.25 \\ 
\end{pmatrix},
\end{equation*}
\normalsize
and let $A_2$ consist of the first eight columns of $A_1$, and $A_3$ of the first four columns of $A_1$. Note that matrix $A_2$ is the same as the one in \cite{cooley2019}.} Let $\nu_i$ denote the exponent measure of a max-linear random vector with coefficient matrix $A_i$, for $i=1,2,3$ and $\alpha = 4$. Let $\Sigma_i = A_i A_i^T$; we obtain the $d! = 5! = 120$ decompositions of  $\Sigma_i$ and consider the 120 max-linear vectors with corresponding coefficient matrices for each $i=1,2,3$.  From the 120 possible paths,  we use {98 paths for $A_1$,  94 for $A_2$,  and 80 for $A_3$}, {since the others paths lead to (near) zero eigenvalues in $\Sigma^{(i)}$; see Remark~\ref{rem:zerosigma}. }

\begin{figure}[h!]
\centering
\includegraphics[width=0.49\textwidth]{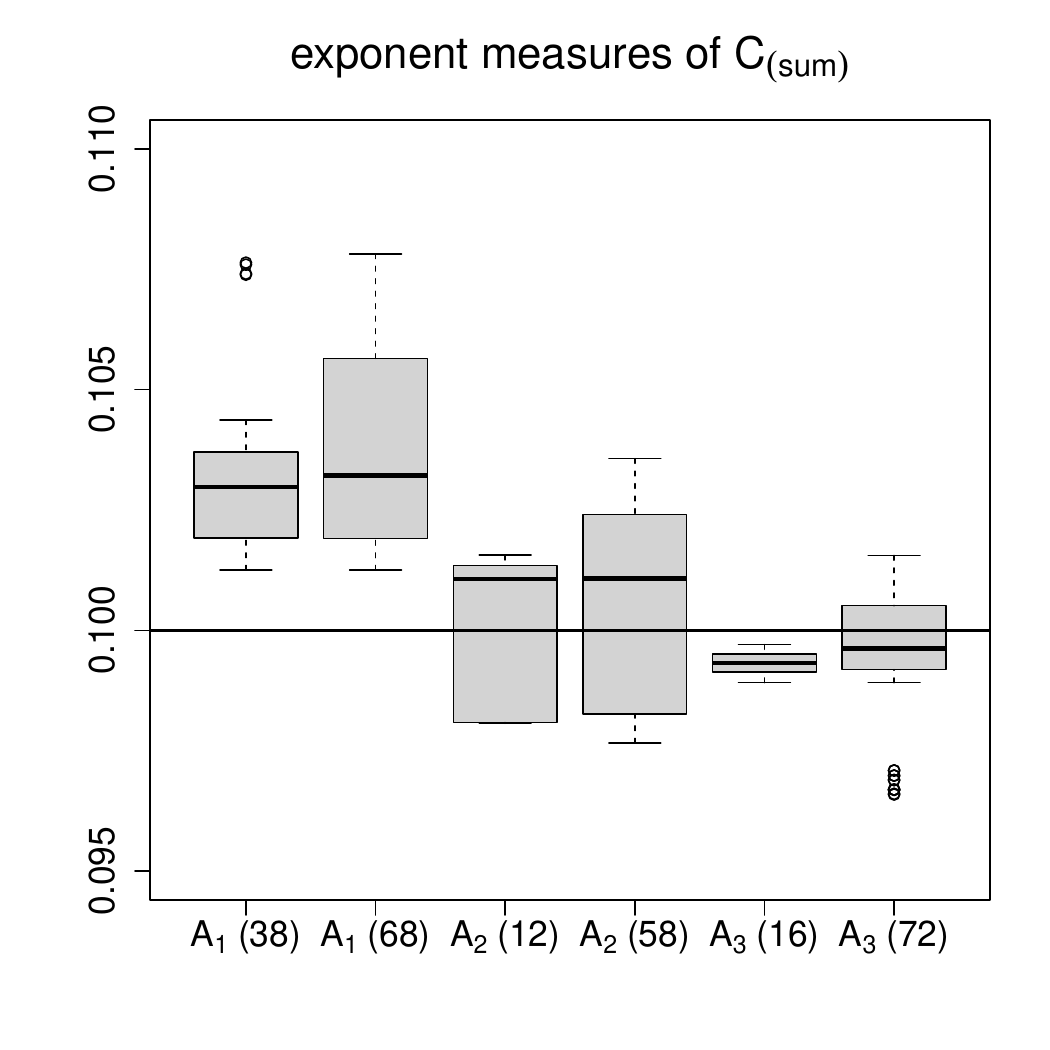} 
\includegraphics[width=0.49\textwidth]{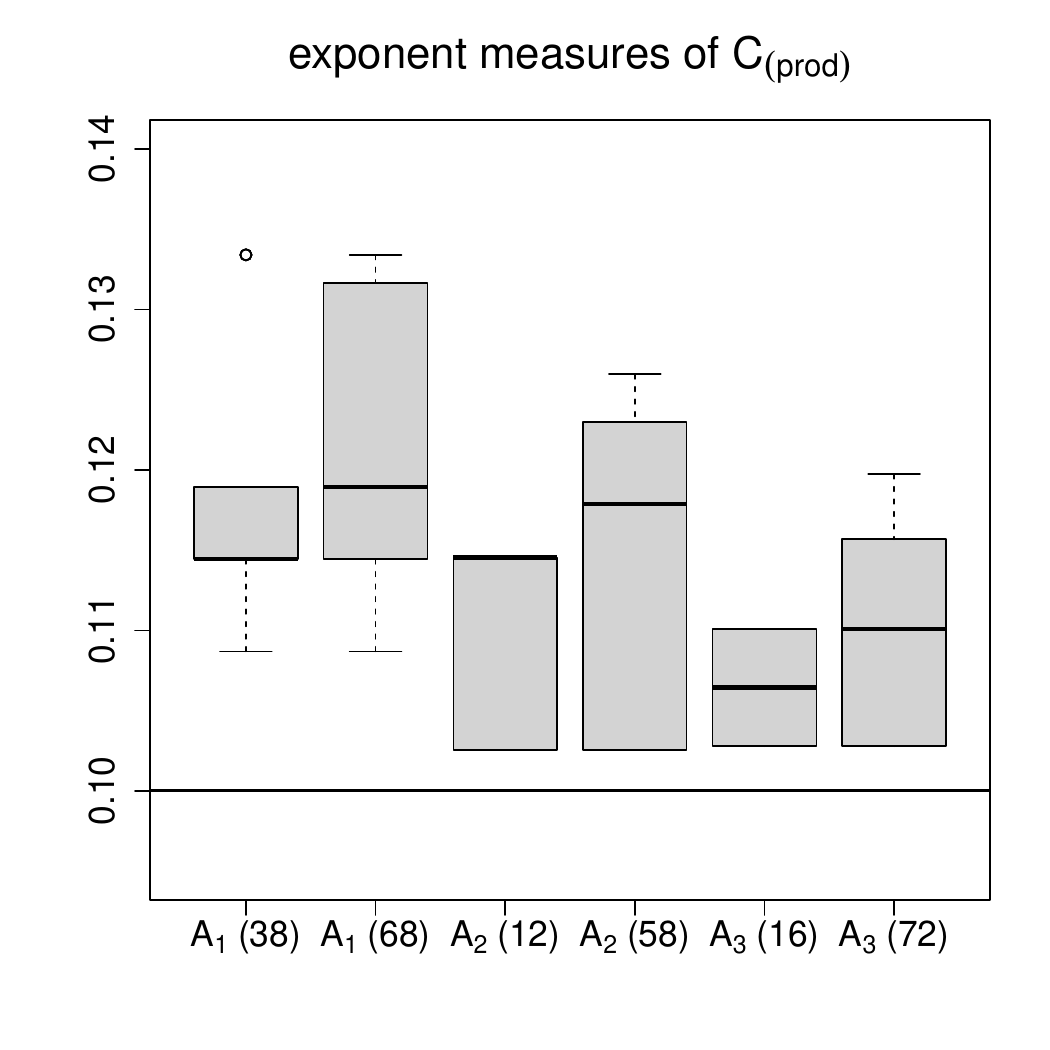} \\
\includegraphics[width=0.49\textwidth]{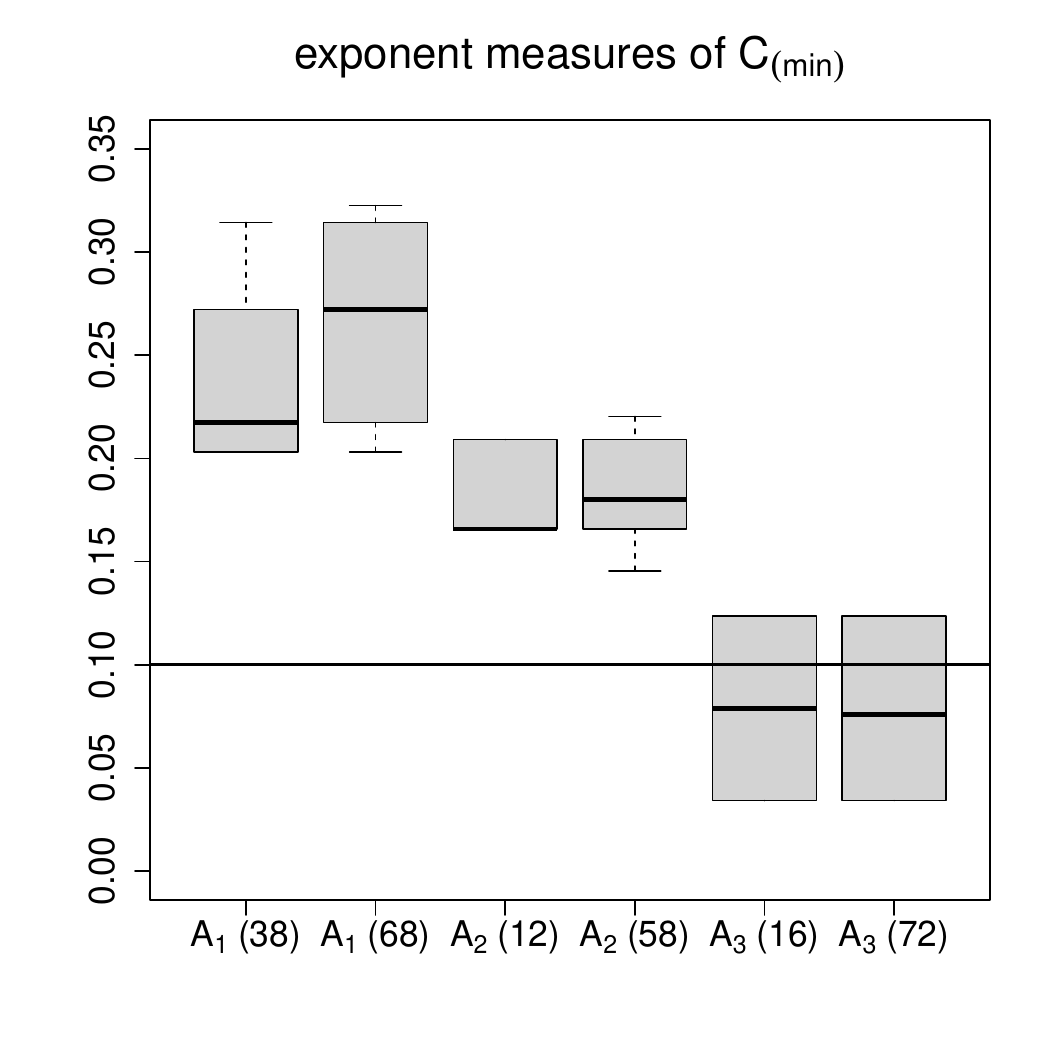}
\includegraphics[width=0.49\textwidth]{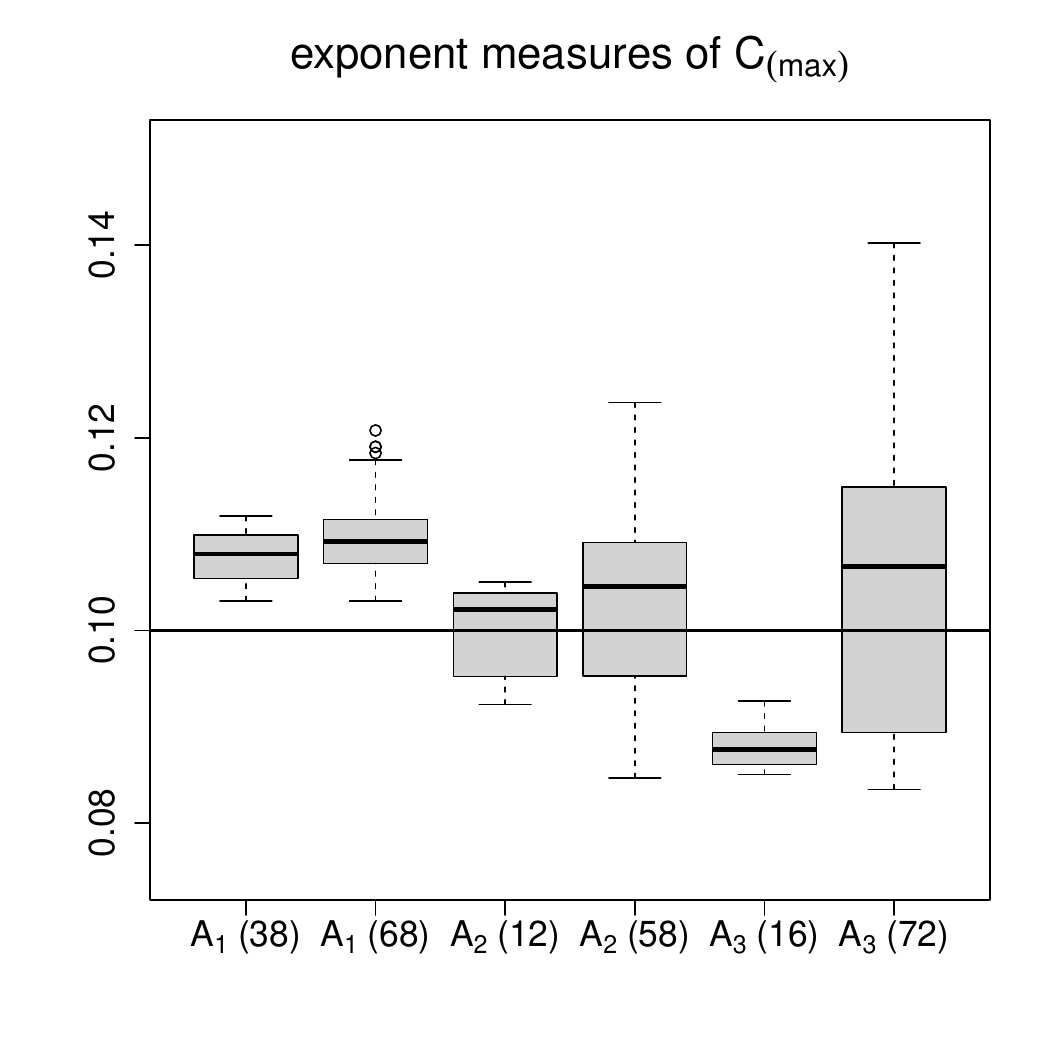} 
\caption{Boxplots of limiting exponent measures of exact and approximate decompositions \change{for $C_{(\textnormal{sum})}
$, $C_{(\textnormal{prod})}$, $C_{(\textnormal{min})}$, and $C_{(\textnormal{max})}$,} and $\alpha = 4$. The numbers in parentheses represent the number of decompositions on which the boxplot is based. The true value based on the full $5$-dimensional dependence model is equal to $0.1$ for all cases.}
\label{fig:expmeasures}
\end{figure}

We assess the limiting exponent measure of $C_f  (x) = \{ \vc{y} \in \EE_0: f(\vc{y}) > x\}$; see Section \ref{sec:failure}. We use the functions \change{$f_{(\textnormal{sum})}(\vy) = \frac{1}{d} \sum_{j=1}^d y_j$, $f_{(\textnormal{prod})} (\vy) = \prod_{j=1}^d y_j$, $f_{(\textnormal{min})} (\vy) = \min(y_1,\ldots,y_d)$ and $f_{(\textnormal{max})} (\vy) = \max(y_1,\ldots,y_d)$, and we write the associated limiting exponent measures as $C_{(\textnormal{sum})}
$, $C_{(\textnormal{prod})}$, $C_{(\textnormal{min})}$, and $C_{(\textnormal{max})}$ respectively.} We consider $\alpha = 4$ to be known and we fix $x$ such that $\nu_i(C_f) = 0.1$.  Figure \ref{fig:expmeasures} shows boxplots for the measures $\nu_i$, $i=1,2,3$, of the sets $C_{(\textnormal{sum})},C_{(\textnormal{prod})}$, $C_{(\textnormal{min})}$ and $C_{(\textnormal{max})}$, obtained through the different paths. 
We show both the results when taking only exact decompositions into account (38 for $A_1$, 12 for $A_2$, and 16 for $A_3$) and when using approximate decompositions, i.e., with Frobenius norm $\leq 5$ (68 for $A_1$,  58 for $A_2$,  and 72 for $A_3$).  
Note that for \change{$C_{(\textnormal{prod})}$ and $C_{(\min)}$}, only the first column of the resulting coefficient matrices is used in the calculation of $\nu_i$, since all subsequent columns contain at least one zero; see also Section \ref{sec:failure}.

Overall, the TPDM \change{provides a reasonable summary of tail dependence: exponent measures calculated based on the decompositions fluctuate around the intended value. We observe that the tail dependence properties of the failure regions $C_{(\textnormal{sum})}$ and $C_{(\max)}$ are better summarized through the TPDM than those of $C_{(\textnormal{prod})}$ and $C_{(\min)}$}. 
This can be explained by the fact that the regions $C_{(\textnormal{sum})}$ and $C_{(\textnormal{max})}$ 
``cut'' the axes and hence concern more scenarios driven by a single extreme event, compared to the regions
$C_{(\textnormal{prod})}$ and $C_{(\textnormal{min})}$ that only concern scenarios involving simultaneous extremes.  In general, variability \change{often increases when using approximate rather than exact decompositions, but the estimates do not deviate more from the true value (on average)}.

\section{Statistical estimation procedure}\label{sec:estim}


\change{In this section, we describe the statistical procedure to estimate the probability of a failure set based on the decomposition algorithm in Section \ref{sec:alg}. We start with the multivariate regular variation case in Section~\ref{sec:estim1}, where we also study the theoretical properties of the estimator. In Section~\ref{sec:estim2}, we discuss the practical implementation of estimating failure probabilities, including the case where multivariate regular variation cannot be assumed.} 


\subsection{Estimation of the TPDM}\label{sec:estim1}
Let $\vc{X}_i = (X_{i1},\ldots,X_{id})$ for $i = 1,\ldots,n$ be iid copies of $\vX$. First, we estimate the tail indices $\alpha_j$, or, equivalently, the extreme value indices $\gamma_j = 1/\alpha_j$, for each $j= 1,\ldots,d$ using the sample $X_{1j},\ldots,X_{nj}$. \change{Recall that multivariate regular variation implies that all margins have equal tail indices, $\alpha_1 = \ldots = \alpha_d := \alpha > 0$.}

Let $X_{(1),j} \leq \ldots \leq X_{(n),j}$ denote the order statistics of $X_{1j},\ldots,X_{nj}$ and let $k = k_n$ be an intermediate sequence such that $k \rightarrow \infty$ and $k/n \rightarrow 0$ as $n \rightarrow \infty$. The \emph{Hill estimator},
\begin{equation}\label{eq:hill}
\widehat{\gamma}_j = \frac{1}{k} \sum_{i=1}^k \log \left( \frac{X_{(n-i+1),j}}{X_{(n-k),j}} \right),
\end{equation}
is the best-known estimator of the extreme-value index \citep{hill1975}, and will be used to obtain the estimated tail indices $\widehat{\alpha}_j = 1/\widehat{\gamma}_j$. 

\change{A single tail index estimate $\widehat{\alpha}$ might be justified: equality of the tail indices can be tested formally using the methods in \citet{chen2024high}. If the null hypothesis $H_0: \alpha_1 = \ldots = \alpha_d$ cannot be rejected,} the entire sample $X_{11},\ldots,X_{n1}, X_{12},\ldots,X_{n2},\ldots,X_{1d},\ldots,X_{nd}$ of length $nd$ will be used to obtain the estimate $\widehat{\alpha}$. \change{The theoretical properties of such an estimator can be derived using the results in \cite{einmahl2022}}.

\change{Let $r_0 = r_0(n)$ be a threshold sequence such that $r_0 \rightarrow \infty$ and $n \, \PP(R > r_0) \rightarrow \infty$ as $n \rightarrow \infty$, and recall from~\eqref{eq:tail of R} that} 
$m=\lim_{n\to\infty} nx^{\alpha} \PP (R>n^{1/\alpha}x)$. For $i  = 1,\ldots,n$, set $R_i =  \norm{\vc{X}_i}_{\widehat{\alpha}}$ and $\vc{W}_i = \vc{X}_i/ R_i$.
{By Statement 1 in Proposition \ref{prop:properties}, we estimate the TPDM by focusing on \change{the angles $\vc{W}_i$ corresponding to observations with high values of the radius component $R_i$}. By replacing $n^{1/\alpha}x$ with $r_0$,} $m$ can be estimated by  $\widehat{m} = r_0^{\widehat{\alpha}}  (n_{\textnormal{exc}}/n)$ for $n_{\textnormal{exc}} = \sum_{i=1}^n \mathbbm{1} \left\{ R_i > r_0 \right\}$.  
{Since $
\sigma_{jk} = m \lim_{x \rightarrow \infty} \EE \left[ W_j^{\alpha/2} W_k^{\alpha/2} \mid R > x \right]$, replacing $x$ by $r_0$,} the elements of the TPDM can be estimated by
\begin{equation}\label{eq:sigmahat}
\widehat{\sigma}_{jk} = \frac{\widehat{m}}{n_{\textnormal{exc}}} \sum_{i=1}^n W_{ij}^{\widehat{\alpha}/2} W_{ik}^{\widehat{\alpha}/2} \mathbbm{1}\left\{ R_i > r_0 \right\} = \frac{r_0^{\widehat{\alpha}}}{n} \sum_{i: R_i > r_0} \left( W_{ij} W_{ik} \right)^{\widehat{\alpha}/2},
\end{equation}
giving $\widehat{\Sigma}_{\vX} =  (\widehat{\sigma}_{jk})_{j,k=1,\ldots,d}$. 

\change{For $i = 1,\ldots,n$, let $\widetilde{R}_i =  \norm{\vc{X}_i}_{\alpha}$ and $\widetilde{\vc{W}}_i = \vc{X}_i/ \widetilde{R}_i$. In the remainder of this subsection, we omit the tildes and write
$R_i$ and $\vc{W}_i$ instead. Before studying the asymptotic behavior of $\hat{\sigma}_{jk}$ in \eqref{eq:sigmahat}, we start by proving the (joint) asymptotic normality of 
\begin{equation}\label{eq:tpdmalpha}
\widetilde{\sigma}_{jk} = \frac{r_0^{\alpha}}{n} \sum_{i: {R}_i > r_0} \left( {W}_{ij} {W}_{ik} \right)^{\alpha/2},
\end{equation}
for all  $1\leq j, k\leq d$, i.e., the case where $\alpha$ is known.} 
\change{A special case of Theorem~\ref{thm:asymptotic_normality}, using $\alpha=2$, has been proven in \citet[Theorem 1]{larsson2012}.}

\change{
\begin{theorem}\label{thm:asymptotic_normality}
    Let $\alpha$ be known and let $r_0=r_0(n)$ be such that $r_0\to\infty$ and $n \, \PP(R>r_0)\to\infty$ as $n\to\infty$. Assume that the following two deterministic relations hold for all $j,k \in \{1,\ldots,d\}$:
\begin{align*}
    \sqrt{n\PP(R>r_0)} \left( r_0^\alpha \, \PP(R>r_0) - m \right) &\to 0,\\
    \sqrt{n\PP(R>r_0)} \left( \mathbb{E}\left[\left( W_{j} W_{k} \right)^{\alpha/2} \mid R>r_0\right] - \frac{\sigma_{jk}}{m} \right) &\to 0.
\end{align*}    
Then, as $n\to\infty$, jointly for all $1\leq j, k\leq d$,
    \[
        \left( \sqrt{n\PP(R>r_0)}\{\widetilde\sigma_{jk} - \sigma_{jk}\}\right)_{1 \leq j,k \leq d} \dto \bm{\zeta},
    \]
where $\bm{\zeta} = (\zeta_{jk})_{1 \leq j,k \leq d}$ is a $d^2$-dimensional mean-zero Gaussian random vector with covariance matrix $S$ indexed by pairs $(j,k)$, with entries
\[
\textnormal{Cov}(\zeta_{j_1,k_1},\zeta_{j_2,k_2})
= s_{(j_1,k_1),(j_2,k_2)}
= m \int_{\mathbb{S}}
(w_{j_1} w_{k_1})^{\alpha/2}
(w_{j_2} w_{k_2})^{\alpha/2}
\, \diff H_{\vX}(\vc{w}).
\]
\end{theorem}
Under multivariate regular variation, as $n\to\infty$ (equivalently, as $r_0 \to \infty$),
\begin{align*}
   r_0^\alpha \, \PP(R>r_0) \to m, \qquad
    \mathbb{E}\left[\left( W_{j} W_{k} \right)^{\alpha/2} \mid R>r_0\right] \to \frac{\sigma_{jk}}{m}.
\end{align*} 
Hence, the assumptions in Theorem~\ref{thm:asymptotic_normality} control the convergence rates of the two limits, ensuring that the corresponding bias terms are negligible compared to the stochastic fluctuations at rate $\{n\PP(R>r_0)\}^{-1/2}$.  Such second-order assumptions are standard in extreme-value analysis and constrain the choice of $r_0$ through the usual bias--variance tradeoff.

In practice, it is natural to choose $r_0$ as a high empirical quantile of
$R_1,\ldots,R_n$. Let $R_{(1)}\leq \cdots \leq R_{(n)}$ denote the order
statistics of $R_1,\ldots,R_n$ and set $r_0=R_{(n-n_{\textnormal{exc}})}$,
where $n_{\textnormal{exc}}\in\{1,\ldots,n-1\}$, i.e.,
there are exactly $n_{\textnormal{exc}}$ observations satisfying $R_i>r_0$.
The next theorem shows that asymptotic normality is preserved under this
choice of $r_0$. 
\begin{theorem}\label{thm:empirical_quantile}
    Let $\alpha$ be known and let $n_{\text{exc}}$ be an intermediate sequence such that $n_{\text{exc}} \to \infty $ and $n_{\text{exc}}/n\to 0$ as $n\to\infty$. Let $r_0= R_{(n-n_{\text{exc}})}$ and let $q(n_{\text{exc}}/n)$ denote the $(1-n_{\text{exc}}/n)$-quantile of $R$. Assume that the following two deterministic relations hold for all $j,k \in \{1,\ldots,d\}$:
\begin{align*}
    \sqrt{n_{\text{exc}}} \left( \frac{n_{\text{exc}}}{n} 
    \{q(n_{\text{exc}}/n)\}^\alpha - m \right) &\to 0,\\
    \sqrt{n_{\text{exc}}} \left( \mathbb{E}\left[\left( W_{j} W_{k} \right)^{\alpha/2} \mid R>q(n_{\text{exc}}/n)\right] - \frac{\sigma_{jk}}{m} \right) &\to 0.
\end{align*}    
Then, as $n\to\infty$, jointly for all $1\leq j, k\leq d$,
    \[
        \left( \sqrt{n_{\text{exc}}}\{\widetilde\sigma_{jk} - \sigma_{jk}\}\right)_{1 \leq j,k \leq d} \dto \bm{\zeta},
    \]
    where $\bm{\zeta}$ is the Gaussian random vector with covariance matrix $S$ given in Theorem~\ref{thm:asymptotic_normality}.    
\end{theorem}
Again, the two assumptions in Theorem~\ref{thm:empirical_quantile} control the bias, and they can be satisfied by choosing $n_{\text{exc}}$ to be compatible with the second-order conditions on the tail of $R$ and the limit behavior of $\vc{W}$.

Finally, under additional conditions, we establish consistency of the estimator $\widehat{\sigma}_{jk}$ in \eqref{eq:sigmahat} for each fixed $(j,k)$. Since $d$ is fixed, this implies uniform consistency over all pairs $(j,k)$. 

\begin{theorem}\label{thm:consistency_estimated_alpha}
    Assume that $\widehat\alpha$ is a consistent estimator of $\alpha$, satisfying
    $\widehat\alpha -\alpha = O_P(s_n)$ for some deterministic sequence $s_n\to 0$. In addition, assume that the threshold $r_0$ is chosen such that $s_n\log r_0\to 0$ as $n\to\infty$. Then, for each fixed pair $(j, k)$, as $n\to\infty$,
    \[
        \widehat{\sigma}_{jk} \xrightarrow{p} \sigma_{jk}.
    \]
\end{theorem}
The condition $s_n \log r_0 \to 0$ ensures that the stochastic error from replacing $\alpha$ by $\widehat{\alpha}$ in the factor $r_0^{\widehat{\alpha}}$ is asymptotically negligible. In practice, we take $\widehat{\alpha}$ as the reciprocal of the Hill estimator \eqref{eq:hill}, which is consistent; see, e.g., \citet[Theorem 3.2.5]{dehaanferreira2006} 

The consistency result in Theorem~\ref{thm:consistency_estimated_alpha}
could in principle be strengthened to asymptotic normality of
$\widehat{\sigma}_{jk}$ by combining a delta-method argument with the
asymptotic normality of $\widehat{\alpha}$. This would require additional
conditions relating the threshold $r_0$ to the number of order statistics
used to estimate $\alpha$. Since the resulting expression is cumbersome and
depends on tuning choices that are difficult to verify in practice, we restrict
attention to consistency.}

\subsection{Estimating failure probabilities}\label{sec:estim2}
Once the TPDM is estimated, we calculate the $d \times d$ matrix $\widehat{A}_*$ obtained by applying the decomposition algorithm outlined in Section \ref{sec:alg} to $\widehat{\Sigma}_{\vc{X}}$.
\change{We remark that there exists a natural benchmark with which we can compare our estimate $\widehat{A}_*$. Consider} the empirical $(d \times n_{\textnormal{exc}})$ estimate 
\[
\widetilde{A}_*= \left(\frac{\widehat{m}}{n_{\textnormal{exc}}} \right)^{1/2} \widehat{W}^{\widehat{\alpha}/2},
\]
where $\widehat{W}^{\widehat{\alpha}/2}$ is the matrix whose columns are the vectors $\vc{W}_i^{\widehat{\alpha}/2}$ for which $R_i > r_0$. Recall that $r_0$ is chosen as a high quantile of the empirical distribution of $R_1,\ldots,R_n$; in practice, we plot different quantiles against the estimates of $m$ and search for a region where estimates are stable. The matrix $\widetilde{A}_* $ is ``inefficient'' in the sense that it has much more columns than $\widehat{A}_*$, especially when $n_{\textnormal{exc}}$ is large.  However, it is easy to calculate and satisfies $\widehat{\Sigma}_{\vX} = \widetilde{A}_*  \widetilde{A}_* ^T$.

\change{After obtaining the two estimators} $\widetilde{A} = ( \widetilde{A}_*)^{2/\widehat{\alpha}}$ and $\widehat{A}= (\widehat{A}_*)^{2/\widehat{\alpha}}$, \change{we use them} to estimate the probability of falling into one of the failure regions presented in Section~\ref{sec:failure}. 
We will denote the estimators of these probabilities by $\widehat{p}$ for $\widehat{A}$ and 
$\widetilde{p}$ for $\widetilde{A} $.
Let $\lVert \, \cdot \, \rVert_F$ denote the Frobenius norm.  
In \change{Section~\ref{sec:appl}}  we will 
compare \emph{exact decompositions} of $\widehat{\Sigma}_{\vX}$, satisfying
$\lVert \widehat{\Sigma}_{\vX}  -  \widehat{A}_{*} \widehat{A}^T_{*} \rVert_F \leq 10^{-12}$ with \emph{approximate decompositions} of $\widehat{\Sigma}_{\vX}$, satisfying
$\lVert \widehat{\Sigma}_{\vX}  -  \widehat{A}_{*} \widehat{A}^T_{*} \rVert_F \leq 5$.

\change{When the estimated tail indices $\widehat{\alpha}_1,\ldots,\widehat{\alpha}_d$ differ too much or whenever (some) tail indices are negative, the multivariate regular variation assumption is not valid. In such cases,} standardizing the margins is necessary. As shown in Proposition~\ref{prop:properties}, the TPDM will be the same regardless of the choice of $\alpha$. 
We set \change{$\alpha=2$ and use the transformation}
\begin{equation}\label{eq:standard}
X_{ij}^{*} = \left( - \log \widehat{F}_j^{P} (X_{ij}) \right)^{-1/2}, \qquad j = 1,\ldots,d, \, i = 1,\ldots,n,
\end{equation}
where $\widehat{F}_j^{P}$ denotes the semi-parametric estimate of $F_j$,
\begin{equation*}
\widehat{F}_j^P (y) = \begin{dcases}
\widehat{F}_j (y) & \text{ if } y \leq u_j ; \\
1 - \left(1 - \widehat{F}_j (u_j) \right) \overline{H}_{\widehat{\sigma}_j, \widehat{\gamma}_j} (y - u_j) & \text{ if } y > u_j; 
\end{dcases}
\end{equation*}
where $\widehat{F}_j$ denotes the empirical distribution function of $X_{1j}, \ldots,X_{nj}$ and $\overline{H}_{\widehat{\sigma}_j,{\widehat{\gamma}}_j}$ the survival function of a generalized Pareto distribution (GPD) with estimated scale $\widehat{\sigma}_j$ and shape $\widehat{\gamma}_j$, obtained by fitting $(X_{ij} - u_j \mid X_{ij} > u_j)_{i=1,\ldots,n}$ to the GPD for some high threshold $u_j$.
The distribution of $X^{*}_{1j}, \ldots, X^{*}_{nj}$ is {approximately} \change{Fr\'echet(2)}. Hence, we have $\sigma_{jj} = 1$ for $j \in \{1,\ldots,d\}$ and $m = H_{\vX} (\mathbb{S}) = d$ does not need to be estimated.  Subsequently, we can decompose the TPDM $\widehat{\Sigma}_{\vc{X}^{*}}$ and construct the associated max-linear model to proceed with estimating the probability of the transformed failure set, \change{either through an analytic formula or via a fast Monte Carlo procedure; see Section~\ref{sec:general}.}


\change{We assess the proposed methodology through an end-to-end simulation study in order to assess whether our procedure yields accurate estimates of probabilities of failure sets. We consider three data-generating processes (DGPs), all in dimension $d = 20$:
\begin{enumerate}
    \item A max-linear model with a square coefficient matrix, obtained by estimating and decomposing the TPDM of the first $d = 20$ industry portfolios discussed in Section~\ref{sec:fip}.
    \item A logistic model with parameter $\theta = 1.7$.
    \item A Brown--Resnick model for $d$ spatial locations, generated uniformly over $[0,1]^2$, with range parameter equal to $0.25$ and smoothness parameter $1$. 
\end{enumerate}
The parameter values for the logistic and Brown--Reesnick model are chosen such that the average strength of tail dependence between pairs (i.e. the average $\chi_{jk}$ over all $j,k = 1,\ldots,d$, $j \neq k$) is the same for all DGPs, around $0.5$.

The failure sets we consider are $C_{(\max)} (x_{\max})$ and $C_{(\textnormal{sum})}(d^{-1} \bm{1}, x_{\textnormal{sum}})$, as defined in Section~\ref{sec:failure}, and
\[
C_{(\geq m)} = \{\vy \in \EE_0 : \sum_{j=1}^d \mathbbm{1} \{y_j > x_{\geq m} \} \geq m \},
\]
i.e., at least $m$ components exceed the high threshold $x_{\geq m}$. We take $m = 5$ throughout. The values of $x_{\max}$, $x_{\textnormal{sum}}$, and $x_{\geq m}$ are chosen such that the failure probability of interest always equals $p = 0.01$. They are obtained by Monte Carlo simulation from the corresponding DGP.

Next, we generate $N_{\textnormal{sim}} = 200$ samples of size $n = 2000$ from each DGP. For each sample, we apply the complete estimation procedure as follows:
\begin{enumerate}
    \item Standardize the margins of the sample to $\alpha = 2$ using~\eqref{eq:standard} and estimate the TPDM using~\eqref{eq:tpdmalpha} with $\alpha = 2$.
    \item Decompose the TPDM following the algorithm in Section~\ref{sec:algorithm}, using the ``pragmatic approach'' described in Section~\ref{sec:practical}, to obtain $\hat{A}$.
    \item Simulate a large sample (of size $10^6$) from a max-linear model with coefficient matrix $\hat{A}$ and apply the inverse of~\eqref{eq:standard} to transform the simulated data back to the original marginal distributions.
    \item Estimate the probability of the failure set empirically as described in  Section~\ref{sec:general}.
\end{enumerate}
Note that we always find an exact decomposition (for all DGPs).

Figure~\ref{fig:simstudy} shows boxplots of the failure probability estimates for each DGP and each failure region. The estimates show good correspondence with the true values, even when the DGP is logistic or Brown--Resnick.

\begin{figure}[ht]
\centering
\includegraphics[width=0.9\textwidth]{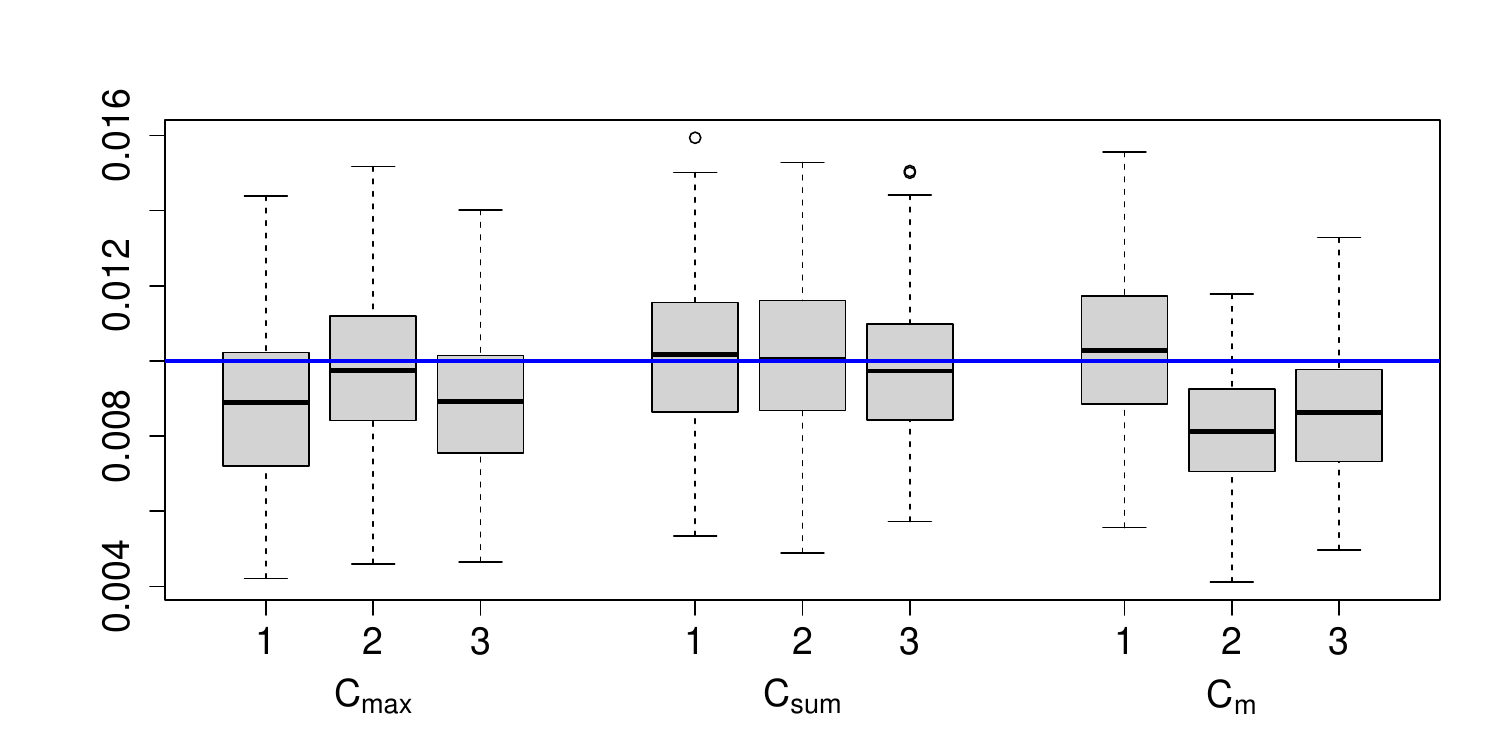}
\caption{Boxplots of estimates of failure probabilities for the regions $C_{(\max)}$, $C_{(\textnormal{sum})}$, and $C_{(\geq m)}$, based on simulations from a max-linear, a logistic, and a Brown--Resnick model.}
\label{fig:simstudy}
\end{figure}

}

\section{Applications}\label{sec:appl}


\subsection{Financial industry portfolios}\label{sec:fip}
We focus on the daily negative value-averaged returns of $d = 30$ industry portfolios,  downloaded from the Kenneth French Data Library.  Data is available for the period January 1970 until December 2019,  leading to a sample size $n = 12 \, 613$. 

We start by estimating the marginal tail indices $\alpha_j$ using the standard Hill estimator with automated threshold selection \citep{danielsson2016}. Figure \ref{fig:FIP} (left) shows the estimates, with upper and lower bootstrap confidence limits for a $95 \%$ confidence level. Overlapping confidence intervals suggest that the hypothesis of a single index of regular variation $\alpha$ is not rejected. The marginal estimates are rather similar, and the straight line corresponding to the estimate $\widehat{\alpha} = 3.478$, based on all data combined, is within all confidence intervals. 
\change{We apply the multiplier bootstrap test of \cite{chen2024high}, designed for the null hypothesis $H_0: \alpha_1 = \ldots = \alpha_d$, based on the marginal threshold values selected above. We find a $p$-value of $0.985$.} 
The next step concerns the choice of $r_0$. Figure \ref{fig:FIP} (right) shows estimates of $m$, the total mass of the angular measure, as a function of $r_0$. We continue with $r_0$ corresponding to the $97.5 \%$ quantile and hence $\widehat{m} = 30$, corresponding to using the 316 largest values.

\begin{figure}[ht]
\centering
\includegraphics[width=0.49\textwidth]{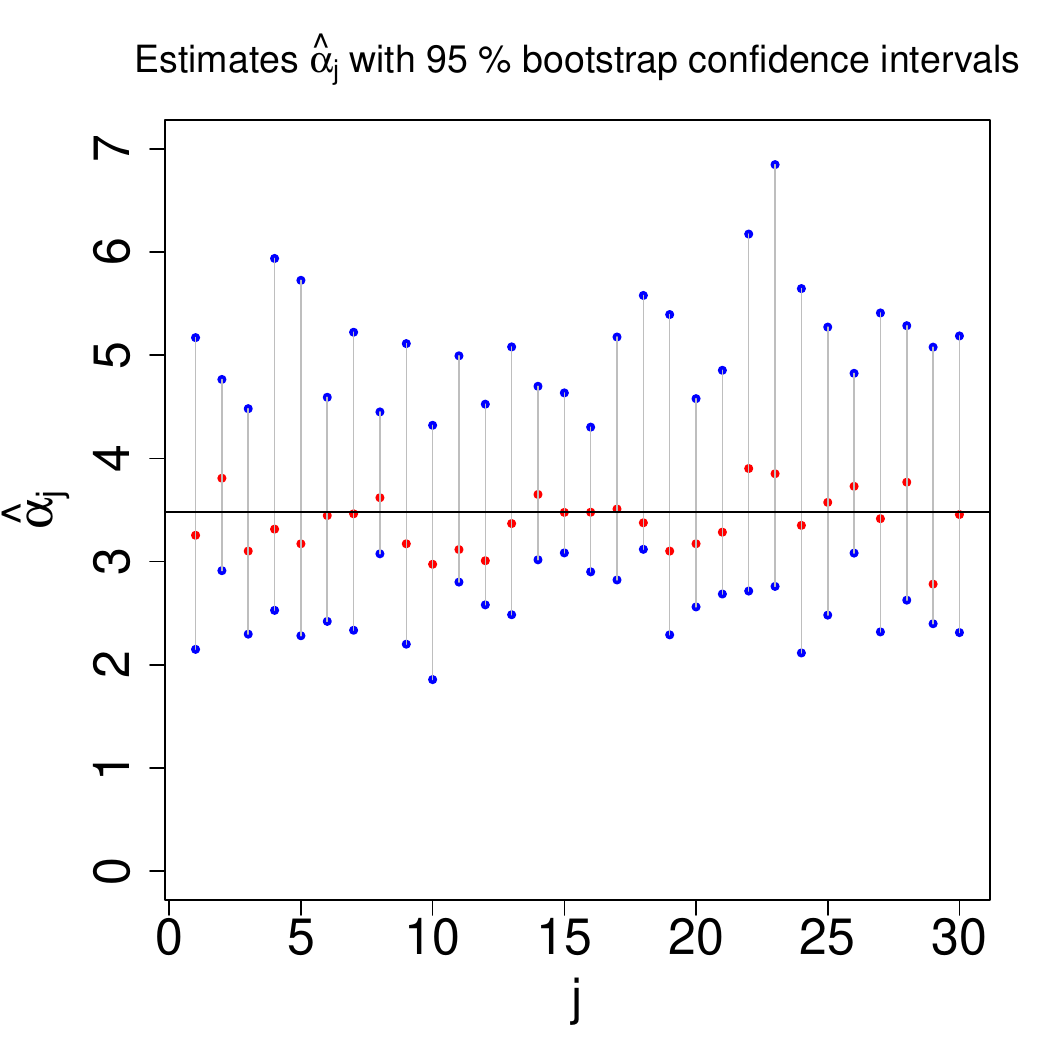}
\includegraphics[width=0.49\textwidth]{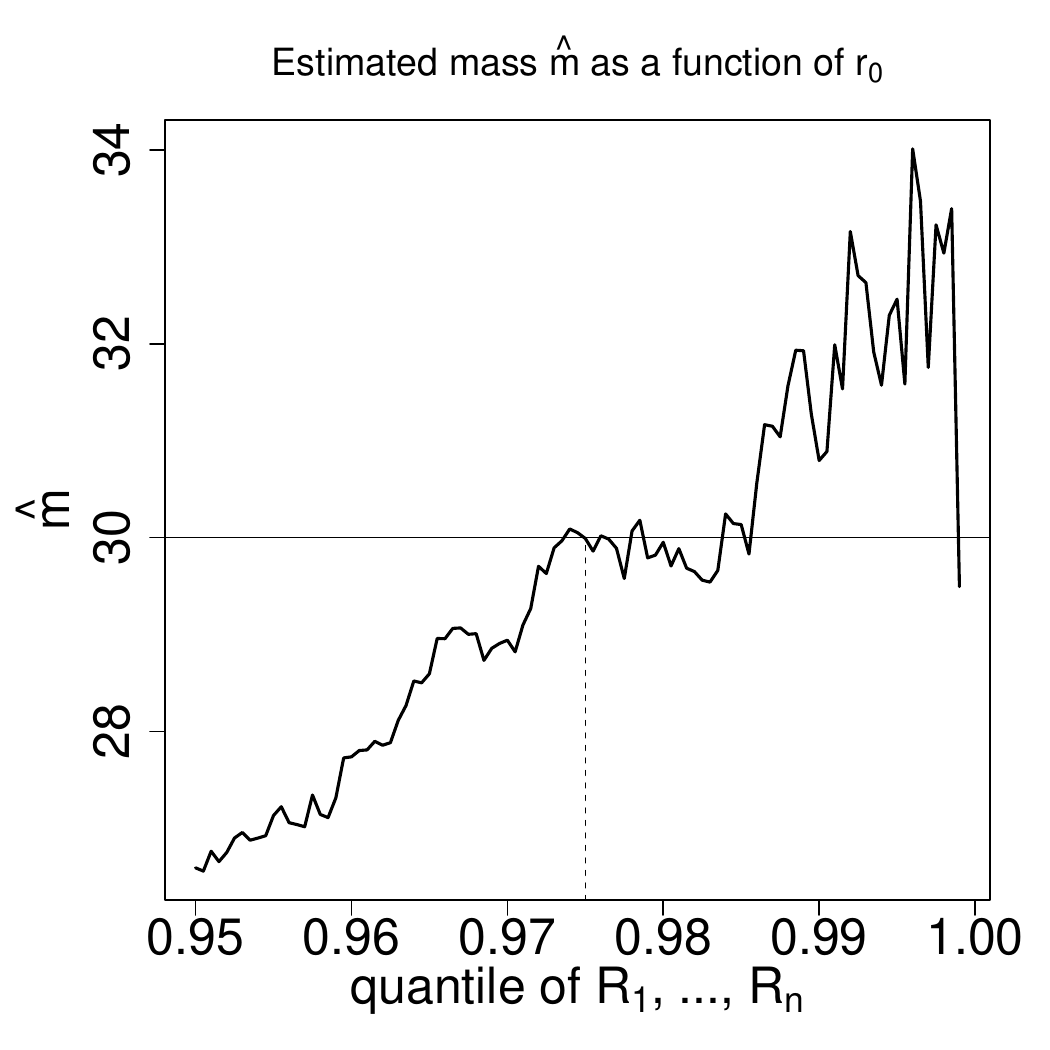}
\caption{Financial industry portfolios. Left: estimates $\widehat{\alpha}_j$ with upper and lower $95 \%$ confidence limits; the horizontal line represents the estimate $\widehat{\alpha}$. Right: estimated mass $\widehat{m}$ of the angular measure as a function of $r_0$, a high quantile of $R_1,\ldots,R_n$; the horizontal line represents the estimate for $r_0 = 0.975$.}
\label{fig:FIP}
\end{figure}

{Finally, we obtain the estimates $\widetilde{A}$ and $\widehat{A}$ based on a (single) decomposition of $\widehat{\Sigma}_{\vX}$, as described in Section \ref{sec:estim}. In order to characterize the uncertainty stemming from the estimation of $\alpha$ and $\Sigma$, we bootstrap the data 100 times and calculate $\widetilde{A}$ and $\widehat{A}$ for each bootstrapped sample (such that $\alpha$ is re-estimated for each sample using the Hill estimator with automated threshold selection, and $m$ is re-estimated for each sample based on a $97.5 \%$ threshold). 
As in Section~\ref{sec:practical}, we present results based on both exact and approximate decompositions. Note that the restriction to approximate decompositions allows for a major time gain (i.e. two minutes for the 100 approximate decompositions versus 24 hours for the exact ones on a standard laptop).
Using $\widetilde{A}_1,\ldots,\widetilde{A}_{100}$ and $\widehat{A}_1, \ldots, \widehat{A}_{100}$, we estimate the following quantities:}
\begin{enumerate}
\item the probability $p_{\summ} = \PP[\vc{X} \in C_{(\textnormal{sum})} (\vc{v}, x)] = \PP[\vc{v}^T \vc{X} > x]$, where $x$ corresponds to the empirical $0.995$ quantile of $(\vc{v} ^T \vX_i)_{i=1,\ldots,n}$; in other words, we validate the empirical $99.5\%$ value-at-risk (VaR) associated with a weighted sum of industry portfolio returns. 
\item the probability $p_{\textnormal{minsum}} = \PP[\vc{X} \in C_{f_1} (x)]$  associated with $C_{f_1} (x) = \{ \vc{y} \in \mathbb{E}_0: f_1(\vc{y}) > x \}$ and
\[
f_1(\vc{y}) = \min \left( \sum_{j=1}^{10} v_j y_j, \sum_{j=11}^{20} v_j y_j, \sum_{j=21}^{30} v_j y_j \right),
\] 
where $x$ is the empirical $0.995$ quantile of $f_1(\vc{X}_1), \ldots, f_1(\vc{X}_n)$. 
\item the probability $p_{\textnormal{maxsum}} = \PP[\vc{X} \in C_{f_2} (x)]$  
associated with $C_{f_2} (x) = \{ \vc{y} \in \mathbb{E}_0: f_2(\vc{y}) > x \}$ and
\[
f_2(\vc{y}) = \max \left( \sum_{j=1}^{10} v_j y_j, \sum_{j=11}^{20} v_j y_j, \sum_{j=21}^{30} v_j y_j \right),
\] 
where $x$ is the empirical $0.995$ quantile of $f_2(\vc{X}_1), \ldots, f_2(\vc{X}_n)$.
\end{enumerate}
{For all three failure regions, we compare equal weights, $\vc{v} = (1/d,\ldots,1/d)$, and unequal weights, $\vc{v} = (0.02, 0.05, 0.03, 0.02, 0.05, 0.03, \ldots, 0.02, 0.05, 0.03)$. All probabilities can be calculated analytically from the estimates of $A$; see Section~\ref{sec:failure}.

Figure \ref{fig2} shows boxplots of the estimates based on $\widehat{A}_1, \ldots, \widehat{A}_{100}$ (exact),  $\widehat{A}_1, \ldots, \widehat{A}_{100}$ (approximate), and $\widetilde{A}_1, \ldots, \widetilde{A}_{100}$. The red dots represent empirical estimates.
\change{Variability is lowest for the sum-region, and highest for the minsum-region. The approximate decompositions yield estimates that are nearly identical to those obtained from exact decompositions.} 


\begin{figure}[t]
\centering
\includegraphics[width=0.33\textwidth]{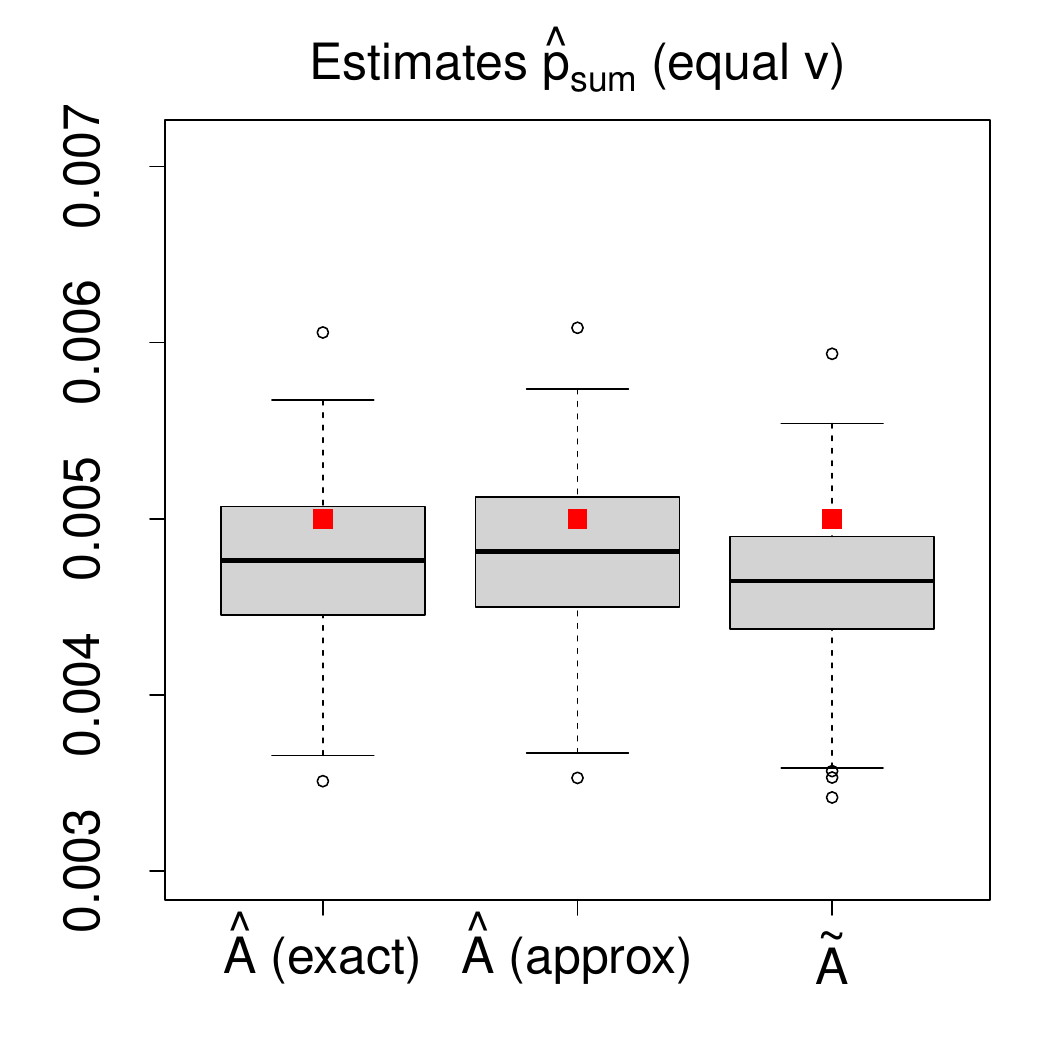} \hspace{-0.2cm}
\includegraphics[width=0.33\textwidth]{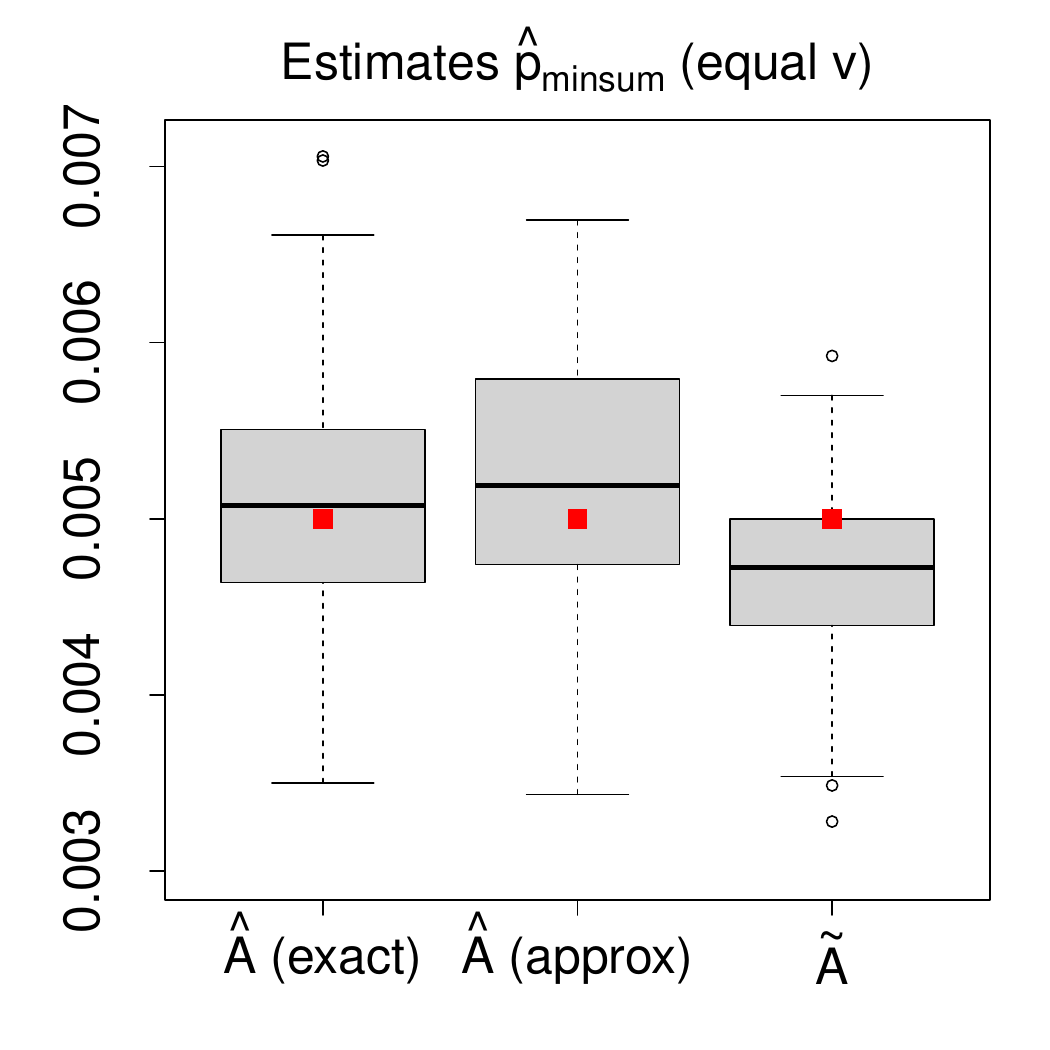}  \hspace{-0.2cm}
\includegraphics[width=0.33\textwidth]{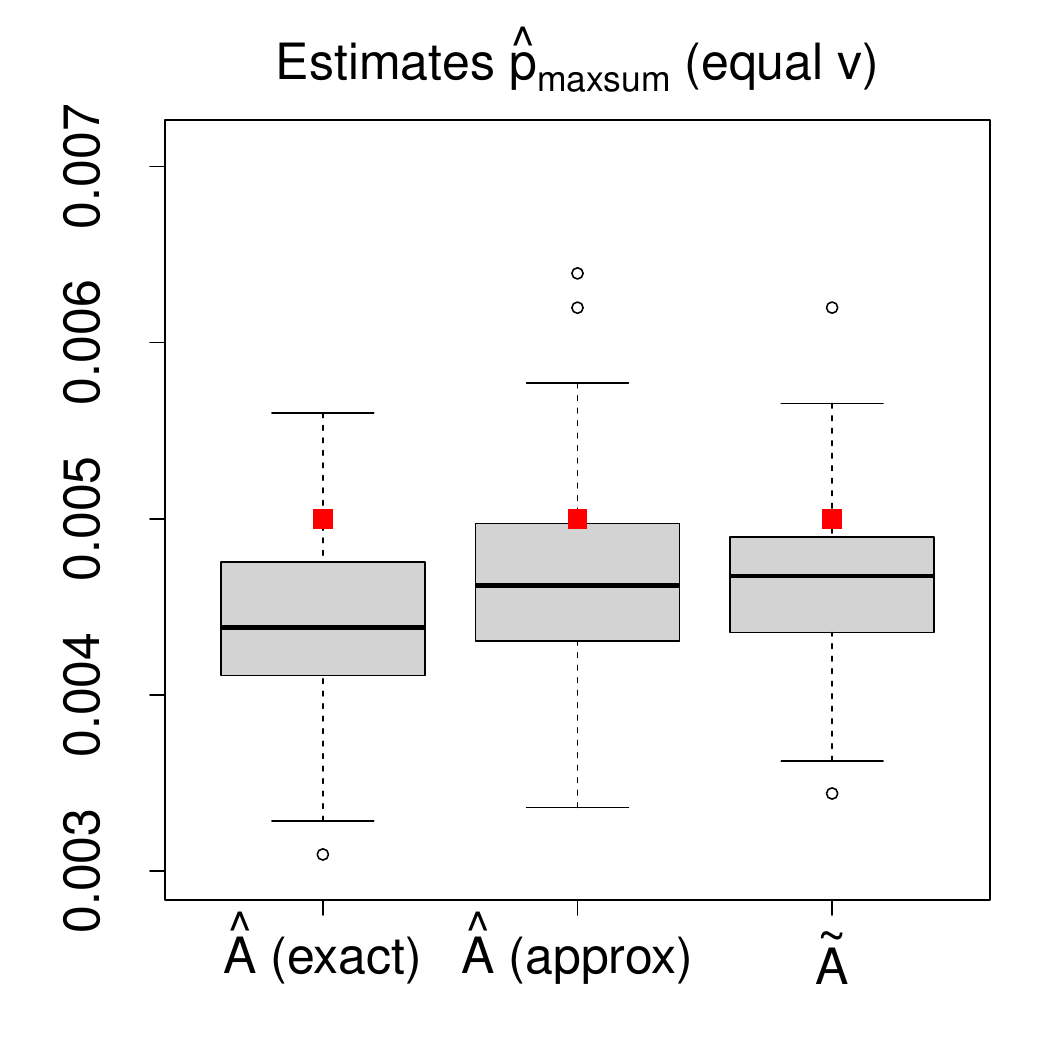}  \\
\includegraphics[width=0.33\textwidth]{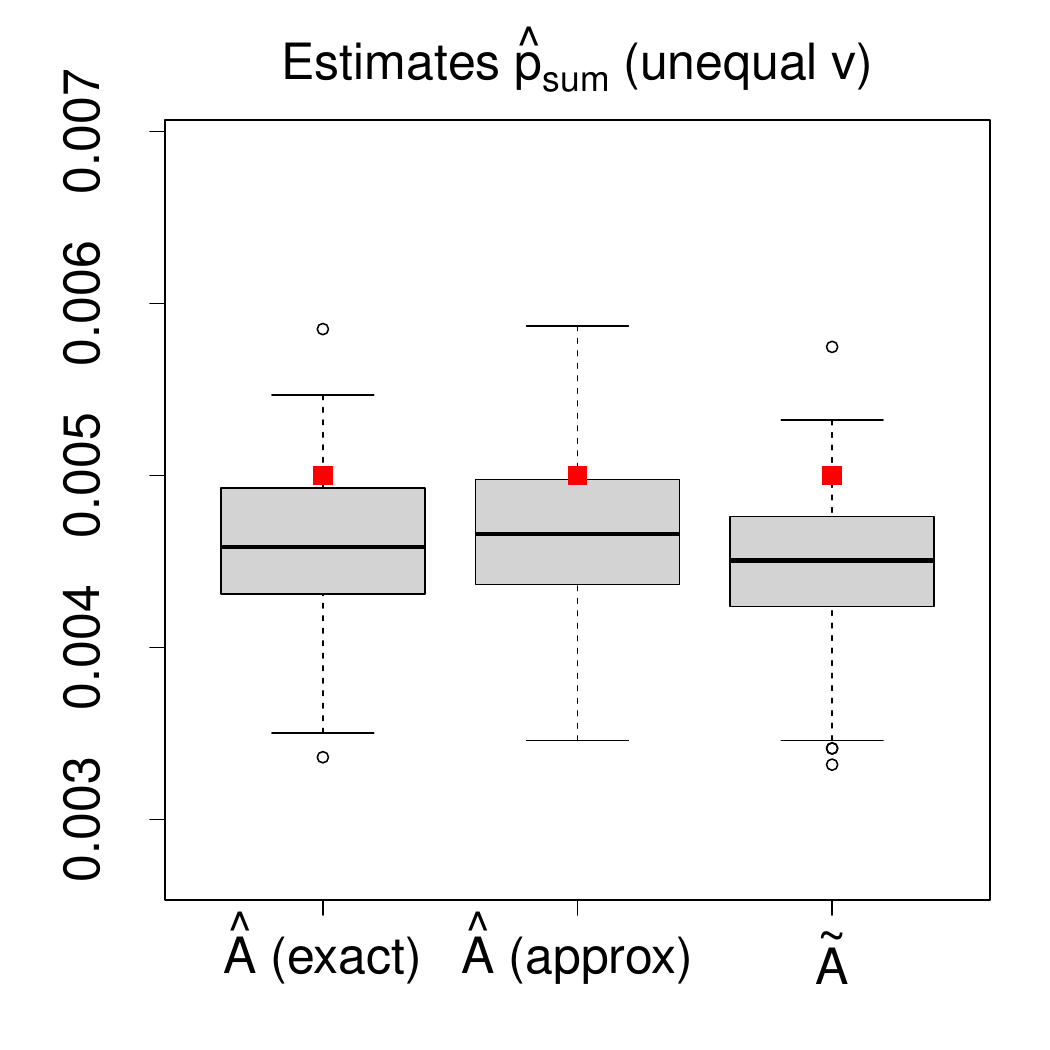} \hspace{-0.2cm}
\includegraphics[width=0.33\textwidth]{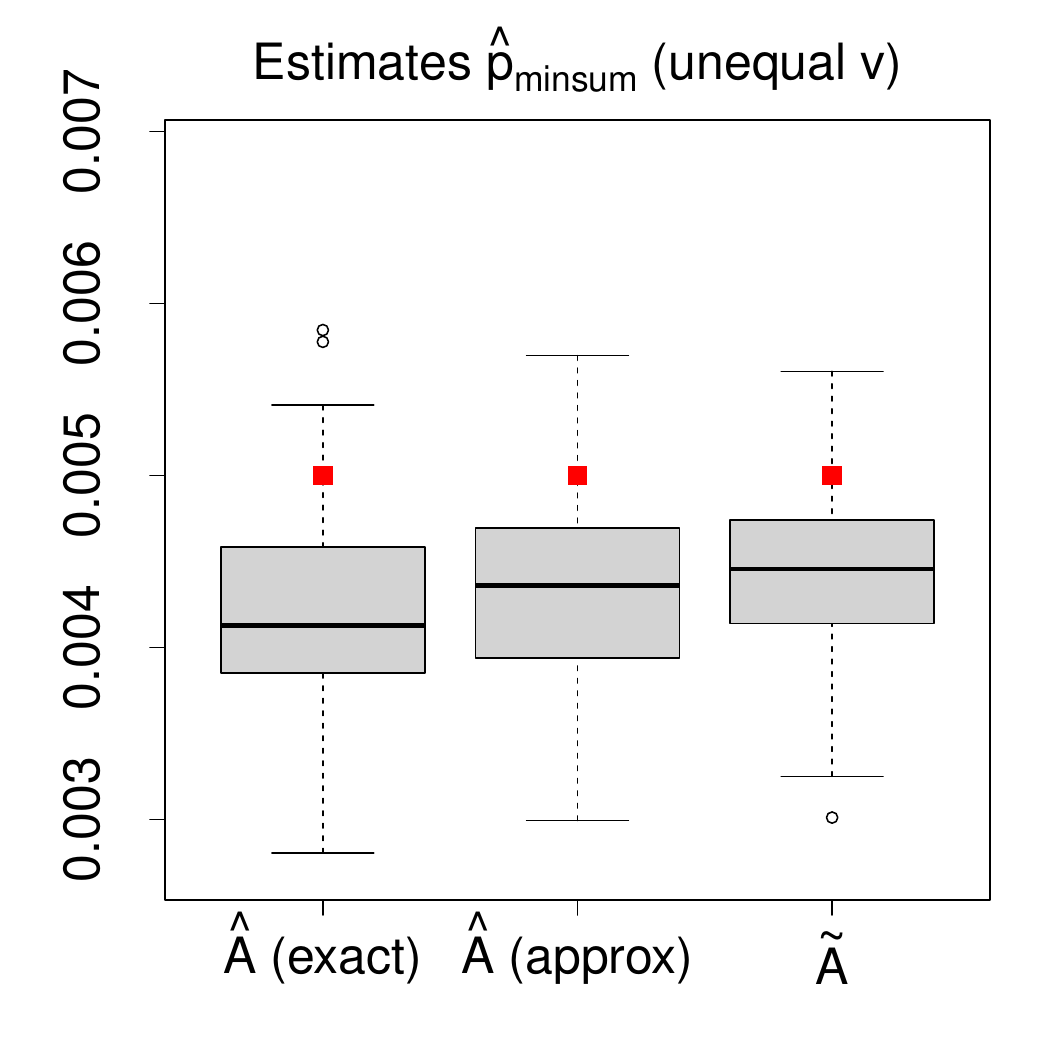}  \hspace{-0.2cm}
\includegraphics[width=0.33\textwidth]{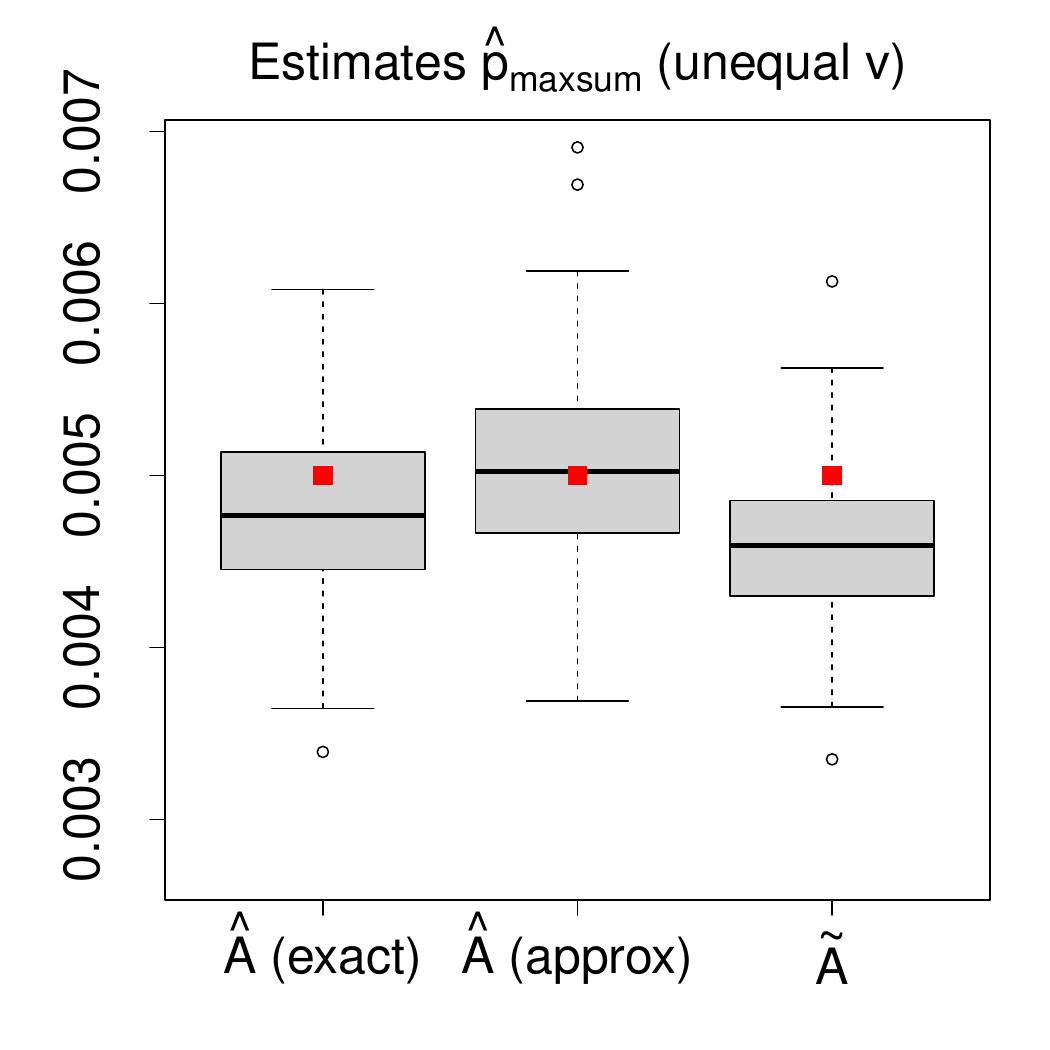} 
\caption{Boxplots of estimates of failure probabilities for the regions $C_{(\textnormal{sum})}$, $C_{f_1}$, and $C_{f_2}$ for the financial industry portfolios, based on equal (top panels) and unequal (bottom panels) weights. Red dots correspond to empirical estimates.}
\label{fig2}
\end{figure}

\subsection{Wind gusts in the Netherlands}\label{sec:wind}
Extreme windstorms are one of the important natural hazards affecting Europe. Impacts are not primarily caused by wind speed itself but by the gusts, i.e., the maxima of the wind speed during a few seconds. 
We consider daily maximal speeds of wind gusts, measured in km/h, observed at $d = 35$ weather stations in the Netherlands during extended winter (October--March), from October 2001 up to and including March 2022.  The data set is freely available from the Royal Netherlands Meteorological Institute (KNMI), see \url{https://climexp.knmi.nl/}.  Please note that up to 50 stations are available in the Netherlands, but 15 were removed because of a large amount of missing data. The remaining 35 stations have between $0$ and $15$ missing data points out of $n = 3827$; these do not occur around clusters of extreme wind gusts and hence they do not impact the following analysis.  
The $35$ weather stations are divided into coastal stations (17) and inland stations (18) based on the same criteria used in \citet{lenderink2009}, because most environmental variables are known to behave differently over land or close to the sea; see Figure \ref{fig:wind2} (left) for a map of the stations used in this study.

\begin{figure}[ht]
\centering
\includegraphics[width=0.33\textwidth]{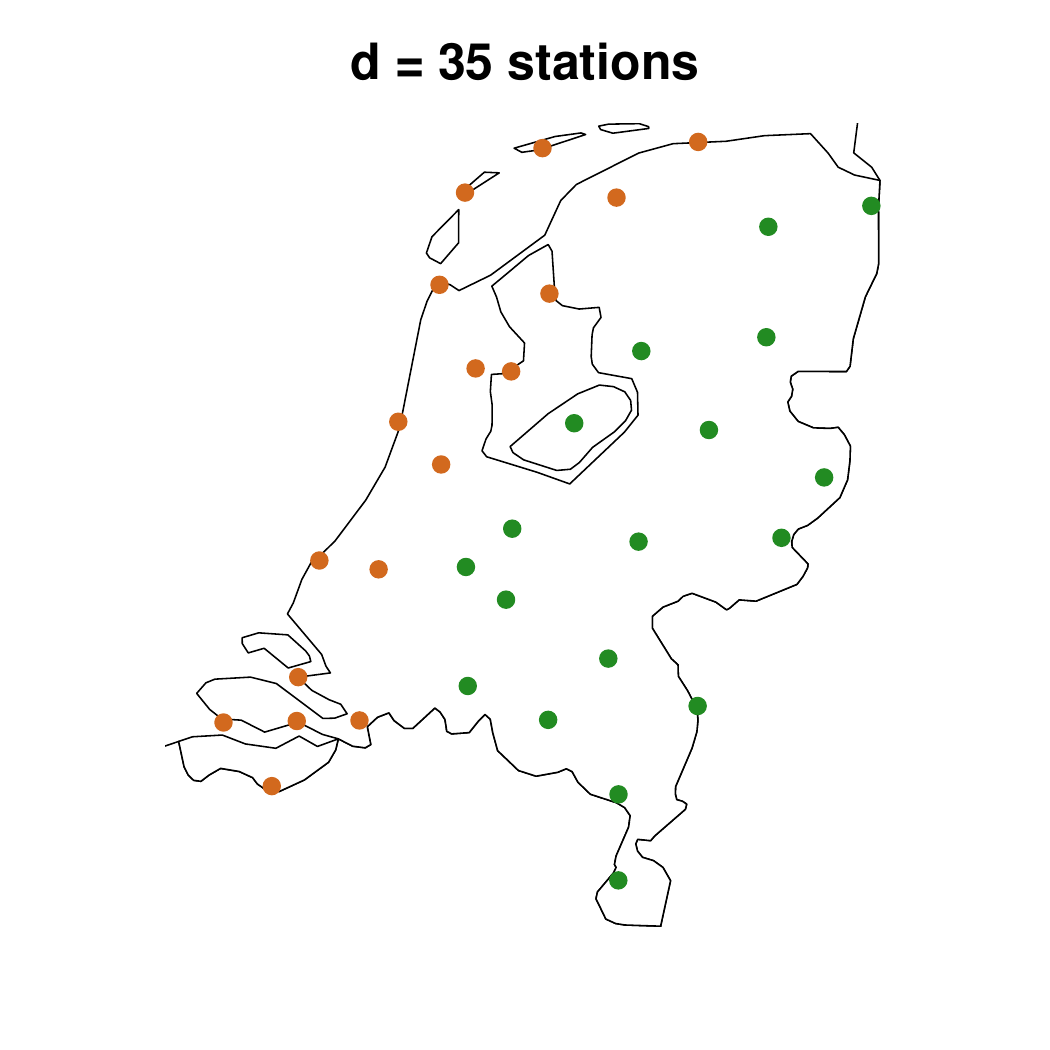} 
\hspace{-0.5cm}
\includegraphics[width=0.33\textwidth]{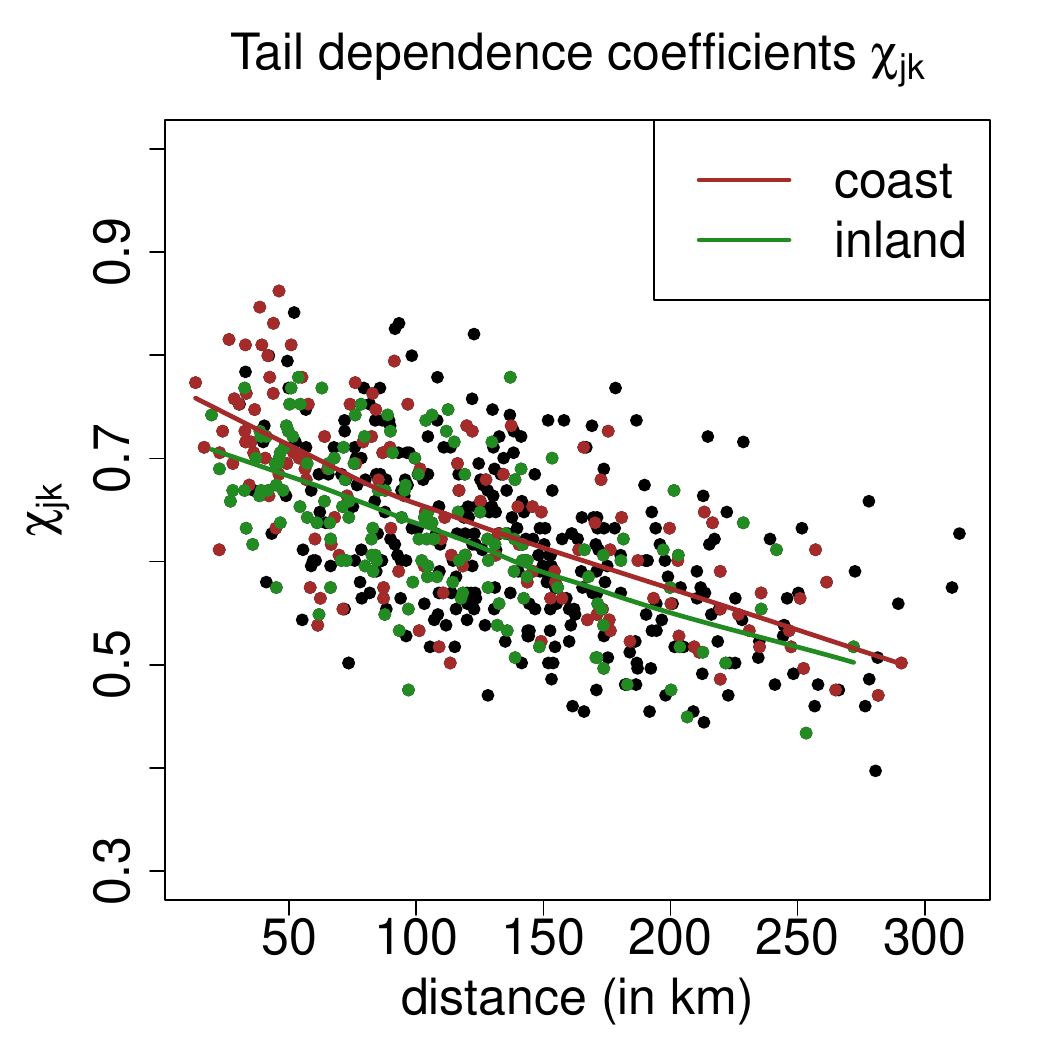} 
\includegraphics[width=0.33\textwidth]{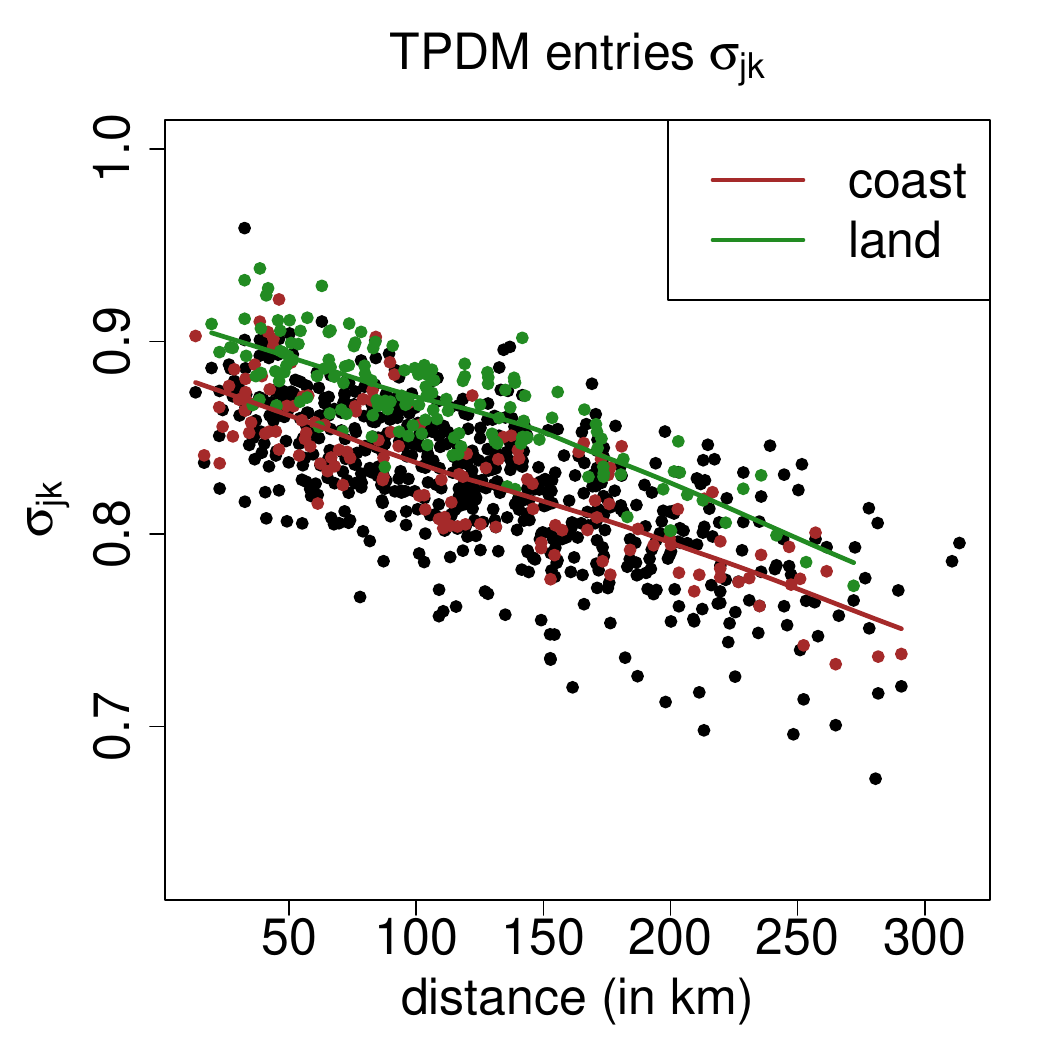}  
\caption{Left: map of the Netherlands with 17 coastal stations in brown and 18 inland stations in green. Middle: empirical coefficients of tail dependence $\chi_{jk}$ for all pairs (black), pairs of coastal stations (brown), and pairs of inland stations (green), with loess curves. Right: same but for the TPDM entries $\sigma_{jk}$.}
\label{fig:wind2}
\end{figure}

{Inspection of the data suggest that the assumption of positive tail indices is not justified. \change{Hence, we standardize the data to Fréchet$(2)$ margins as described in Section~\ref{sec:estim}, and we continue with $\vX^*_1,\ldots,\vX^*_n$.} 
Next, we estimate pairwise tail dependence coefficients non-parametrically to examine the difference in pairwise dependence structure between the coastal and the inland stations. Figure \ref{fig:wind2} shows both the pairwise tail dependence coefficients $\chi_{jk}$ (middle) and the tail pairwise dependence matrix entries $\sigma_{jk}$ (right), for $j,k =1,\ldots,d$, {both calculated based on a 95 \% threshold}. The plot suggests that there is no significant difference between pairs of inland or pairs of coastal stations, and that the wind gusts are asymptotically dependent, even for large distances.

Finally, we obtain the estimates $\widetilde{A}$ and $\widehat{A}$ based on a (single) decomposition of \change{the TPDM estimated based on $\vX^*_1,\ldots,\vX^*_n$}, as described in Section \ref{sec:estim}. In order to characterize the uncertainty stemming from the \change{marginal} standardization and the estimation of $\Sigma$, we bootstrap the data \change{$\vX_1,\ldots,\vX_n$} (on the original scale) 100 times, \change{standardize each bootstrapped sample to $\alpha = 2$, and calculate $\widetilde{A}$ and $\widehat{A}$}. 

We present results for exact decompositions only because the dimensions are sufficiently low to ensure quick calculations.
Then, using $\widetilde{A}_1,\ldots,\widetilde{A}_{100}$ and $\widehat{A}_1, \ldots, \widehat{A}_{100}$, we estimate
\change{
\begin{enumerate}
\item the probabilities $p_{(\max)} = \PP[\vc{X} \in C_{\max} (x \vc{1}_d)]$ for $x \in \{100, 120, 130\}$. The first two values are the thresholds for a code orange and a code red alarm respectively, issued by the KNMI during winter season.
\item the probabilities $p_{(\textnormal{sum})} =
\PP [d^{-1} \sum_j X_j > x]$ for $x \in \{80, 90, 100\}$, quantifying the risk that averaged maximum wind gusts (inland or coastal) reach extreme levels.
\end{enumerate}
Note that for $p_{(\max)}$,  we can find $\bm{x}^*$ such that 
$\left\{\bm{X}\in C_{\max} (x \vc{1}_d)\right\} = \left\{\bm{X}^*\in C_{\max} (\bm{x}^* )\right\}$
by component-wise marginal transformation. On the other hand, to calculate $p_{(\textnormal{sum})}$, we simulate $N^* = 100 \, 000$ observations from max-linear models with coefficient matrices $\widetilde{A}_1,\ldots,\widetilde{A}_{100}$ and $\widehat{A}_1, \ldots, \widehat{A}_{100}$, transform the margins back to the original scale, and estimate $p_{(\textnormal{sum})}$ empirically; see also Section~\ref{sec:general}}.


Figure \ref{fig:wind3} shows boxplots of the estimates based on {$\widehat{A}_1, \ldots, \widehat{A}_{100}$ and $\widetilde{A}_1, \ldots, \widetilde{A}_{100}$} for the three thresholds (from left to right). \change{The top two rows present $p_{(\max)}$ (coastal and inland), the bottom two rows $p_{(\textnormal{sum})}$ (coastal and inland).} Red dots represent empirical estimates. Even though no apparent difference between coastal and inland stations was visible based on pairwise coefficients, when estimating failure probabilities based on the $17$- and $18$-dimensional datasets, we observe that values for coastal stations are up to \change{six times as high as for inland stations, for both failure regions}. In general, the estimates based on $\tilde{A}$ show a slightly larger discrepancy with the empirical estimates than the estimates based on the completely positive decomposition.} 

\begin{figure}[p]
\centering
\includegraphics[width=0.33\textwidth]{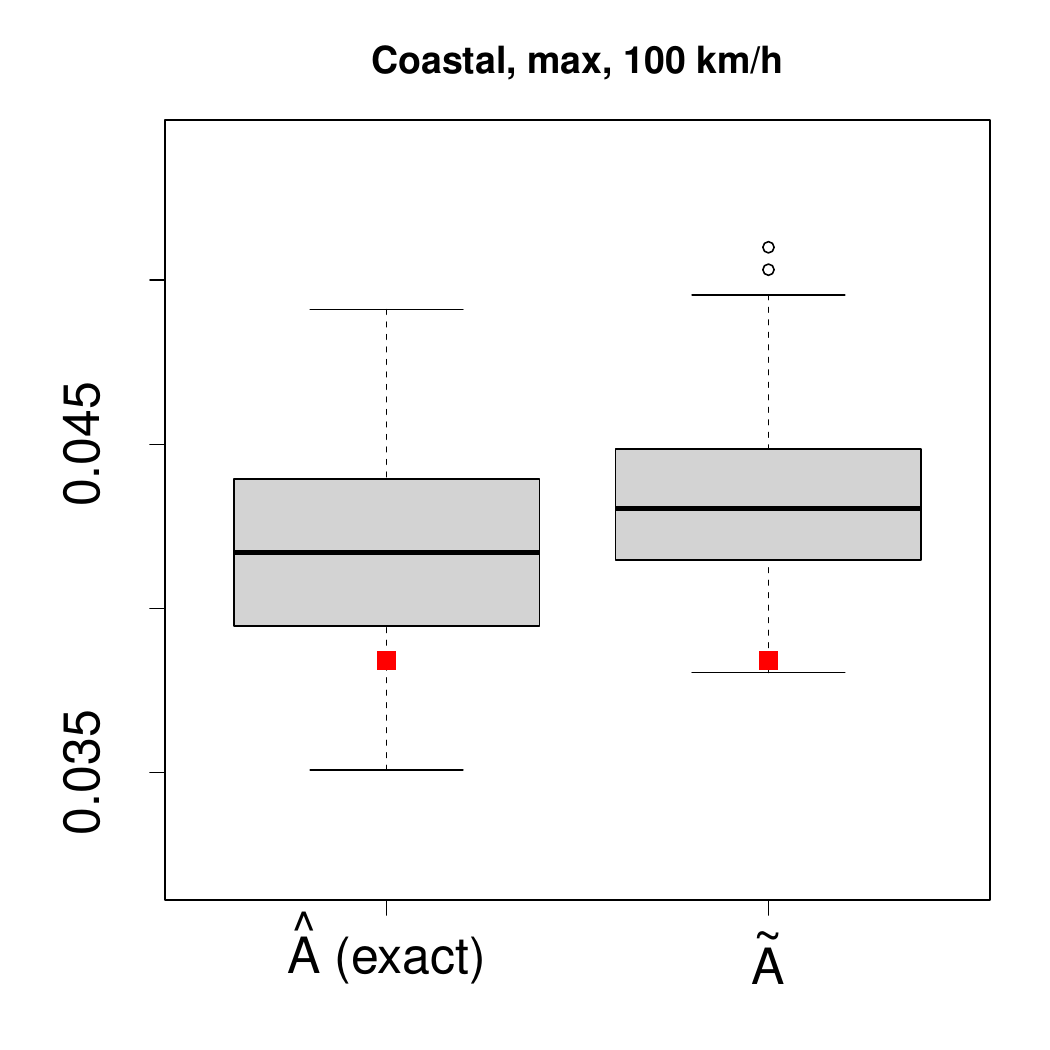} \hspace{-0.25cm}
\includegraphics[width=0.33\textwidth]{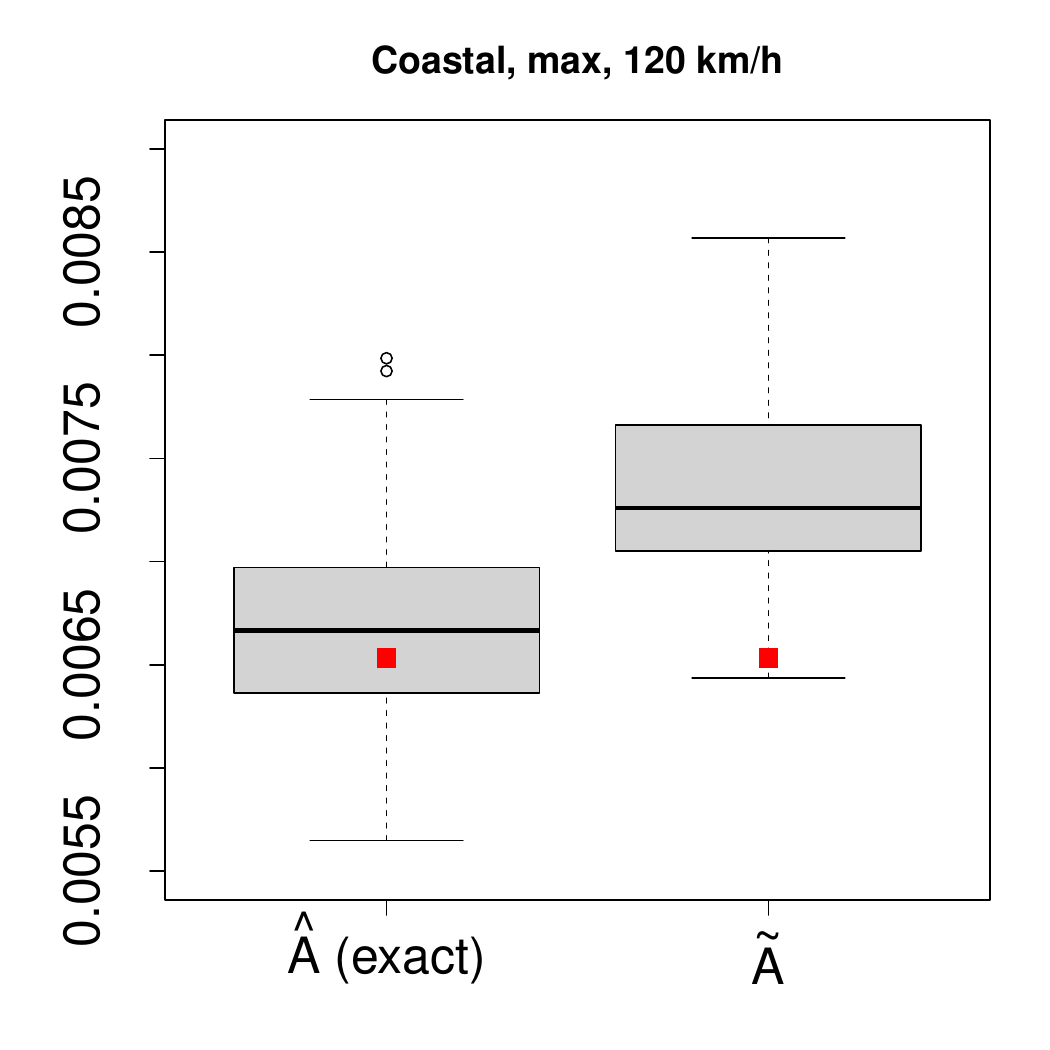}  \hspace{-0.25cm}
\includegraphics[width=0.33\textwidth]{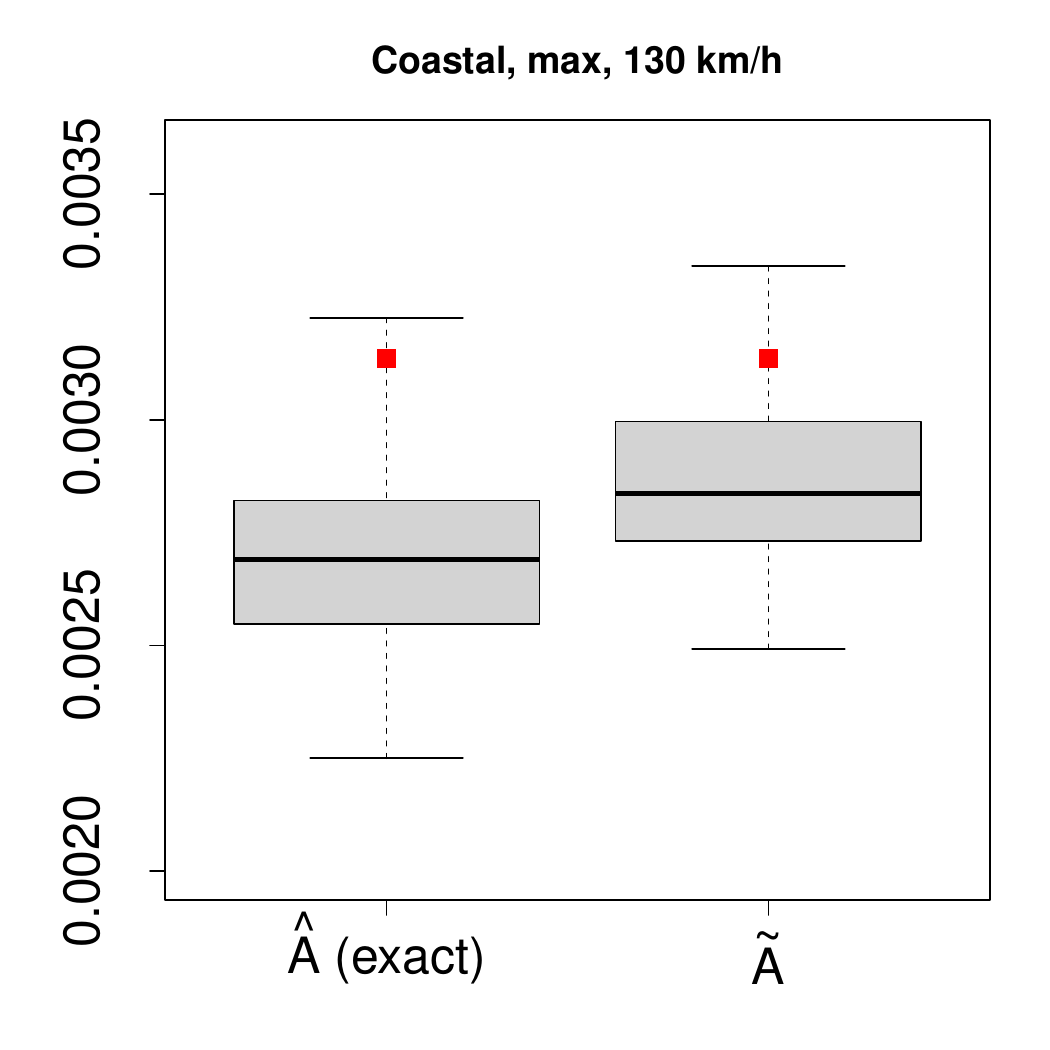}  \\
\vspace{-0.25cm}
\includegraphics[width=0.33\textwidth]{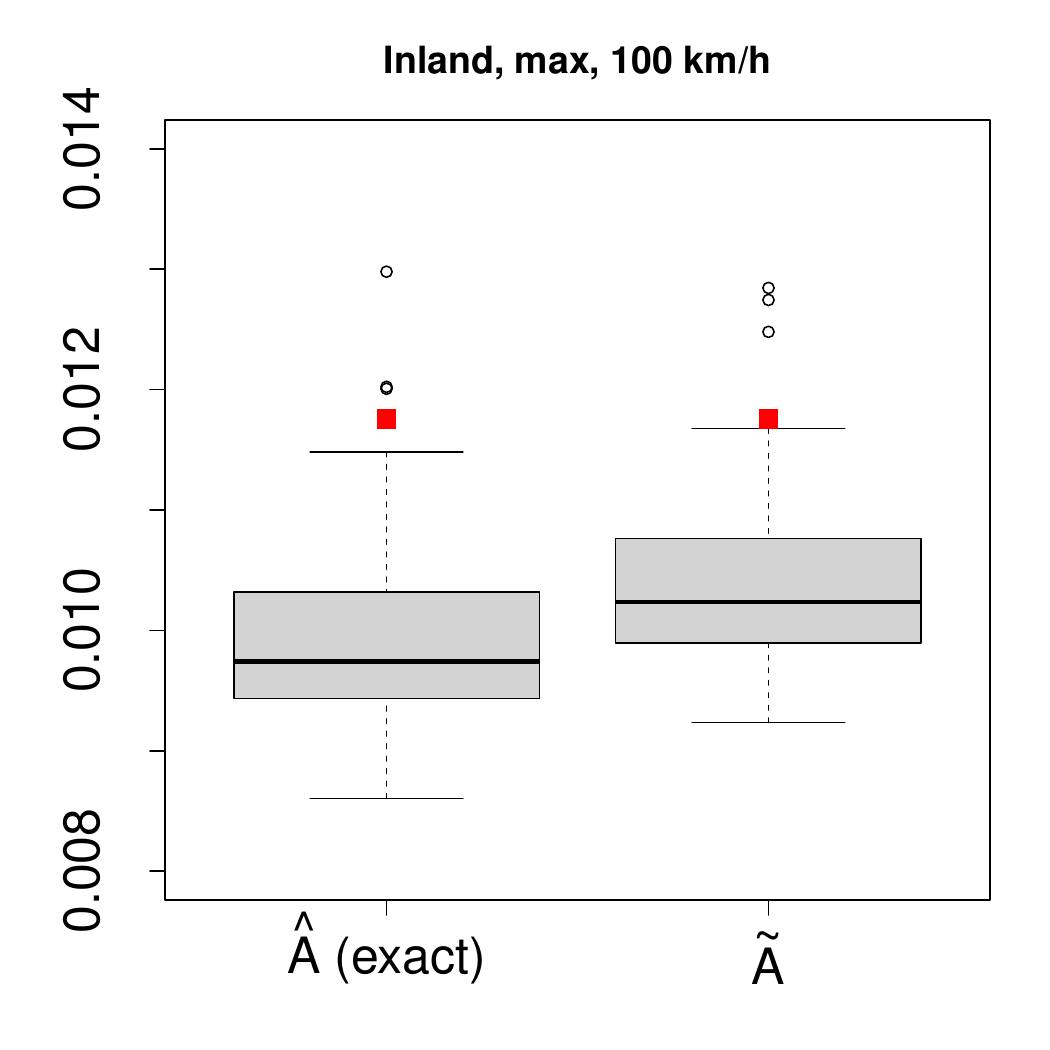} \hspace{-0.25cm}
\includegraphics[width=0.33\textwidth]{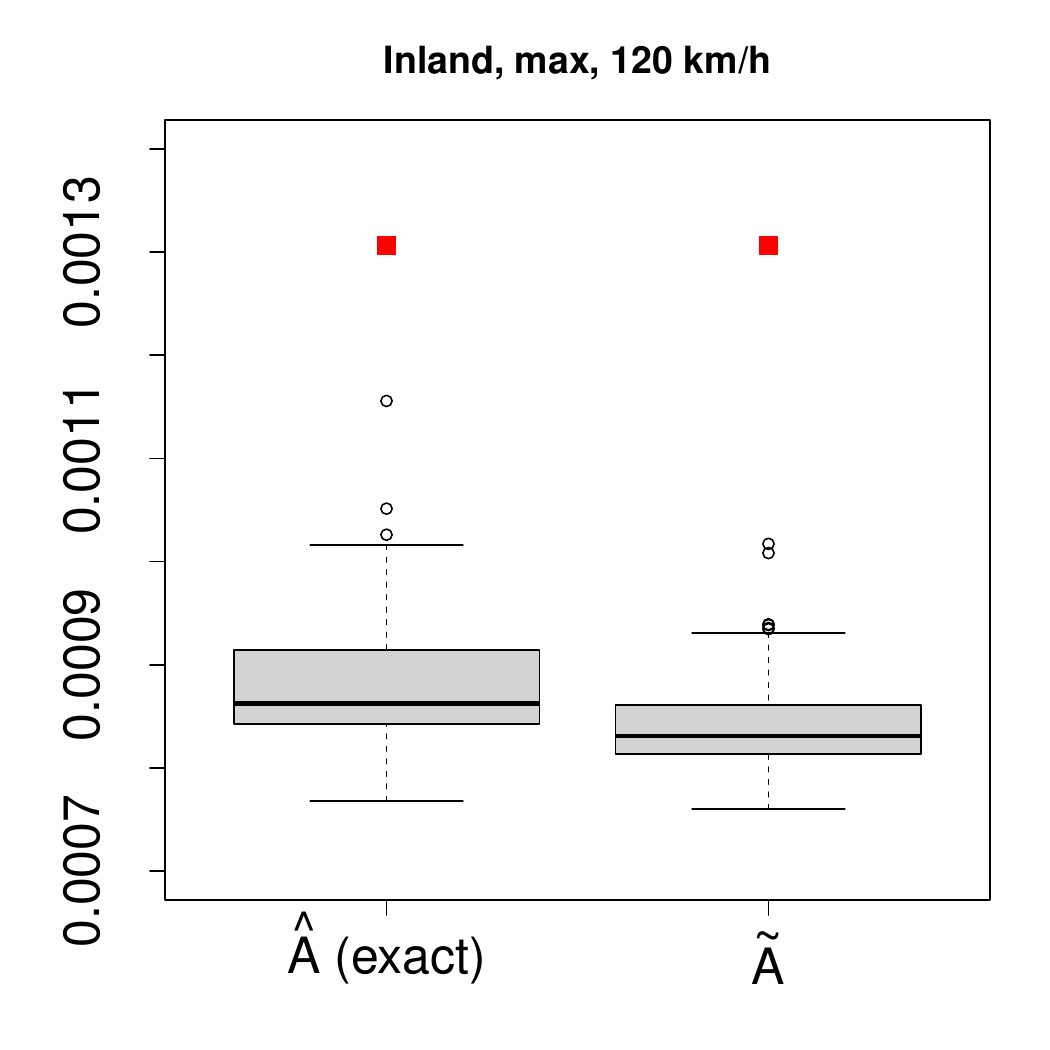}  \hspace{-0.25cm}
\includegraphics[width=0.33\textwidth]{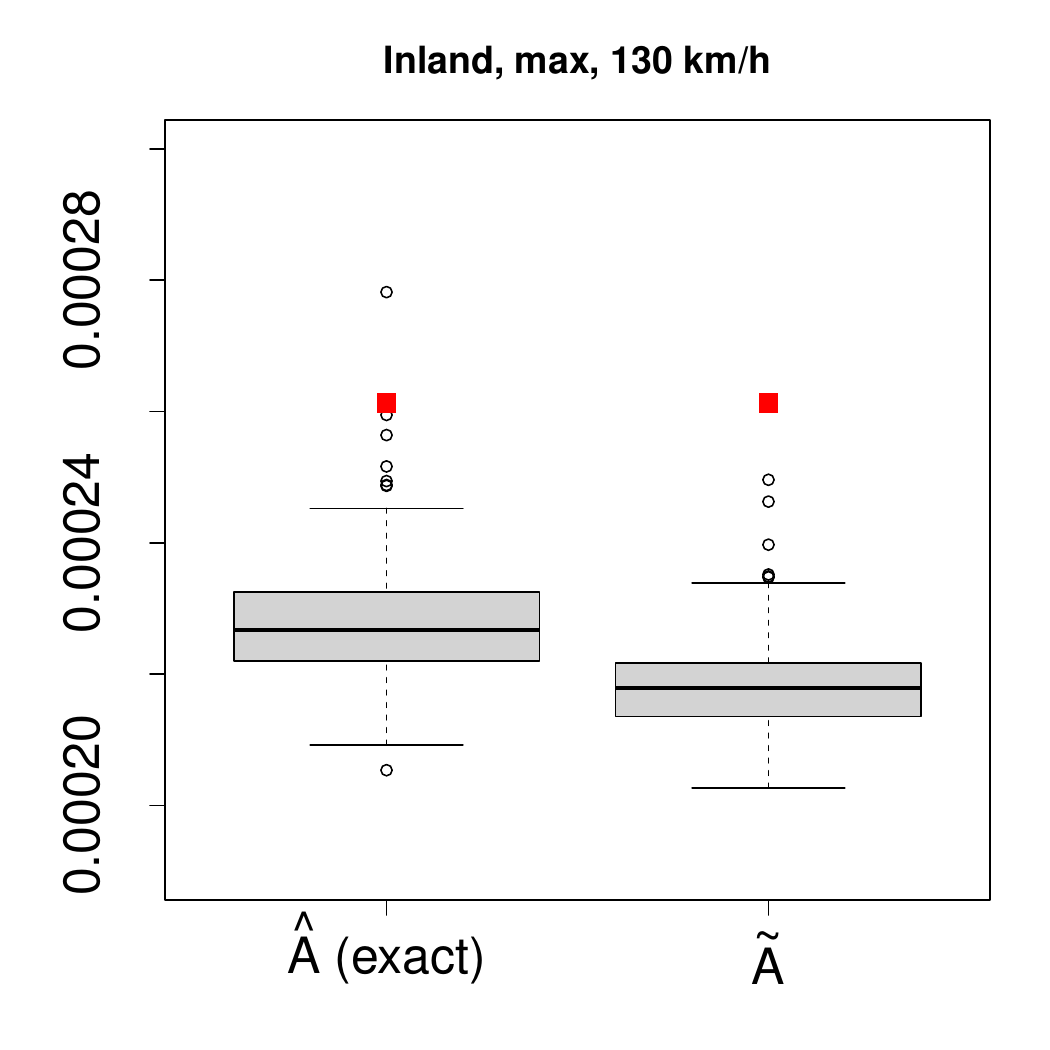}  \\
\vspace{-0.25cm}
\includegraphics[width=0.33\textwidth]{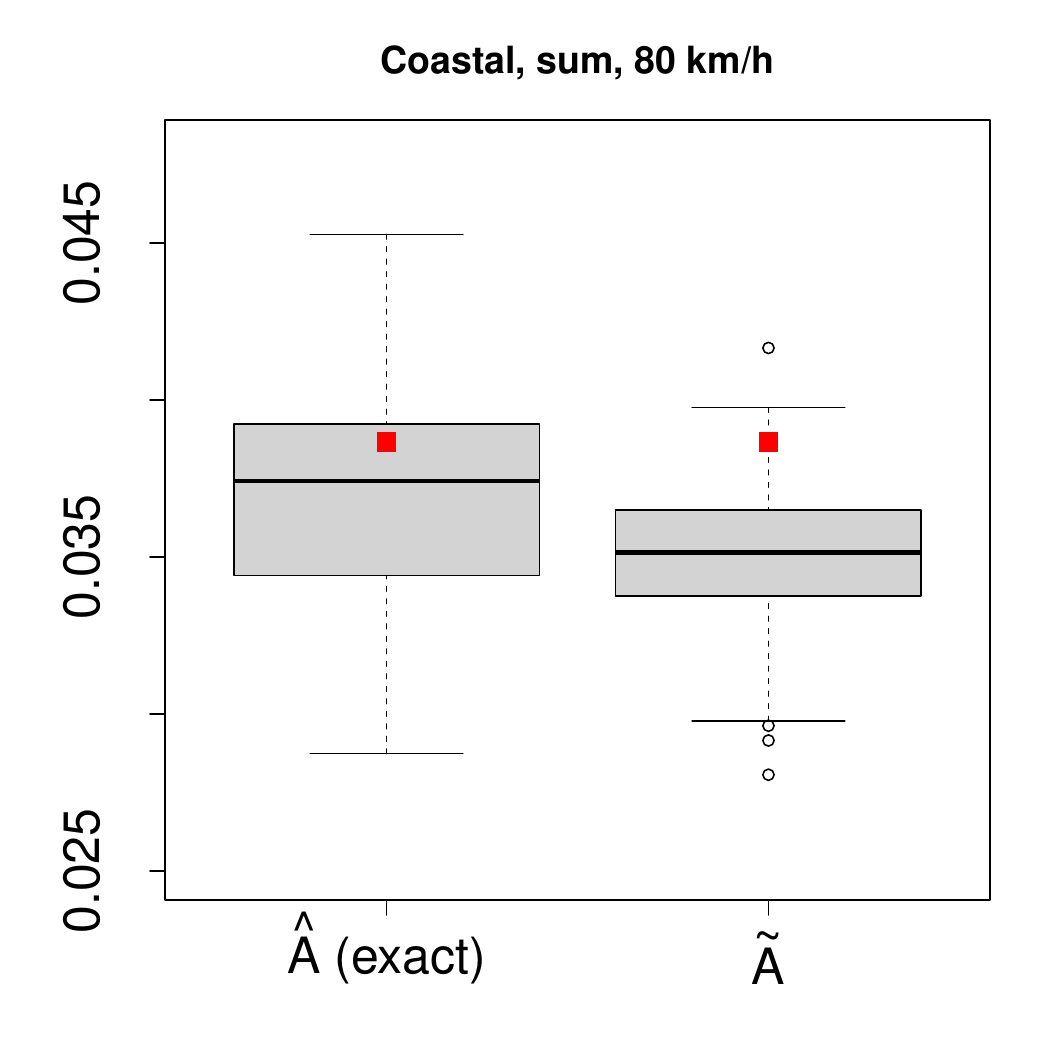} \hspace{-0.25cm}
\includegraphics[width=0.33\textwidth]{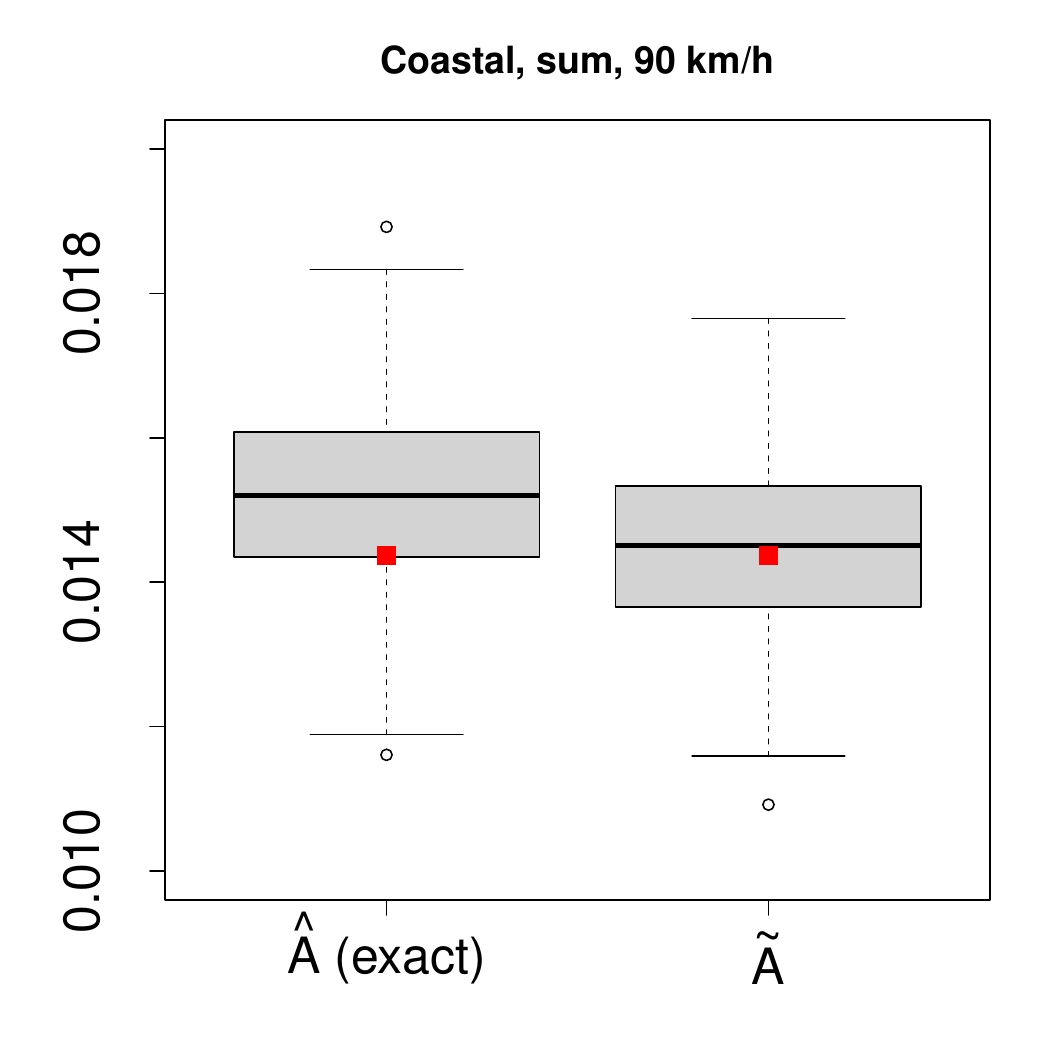}  \hspace{-0.25cm}
\includegraphics[width=0.33\textwidth]{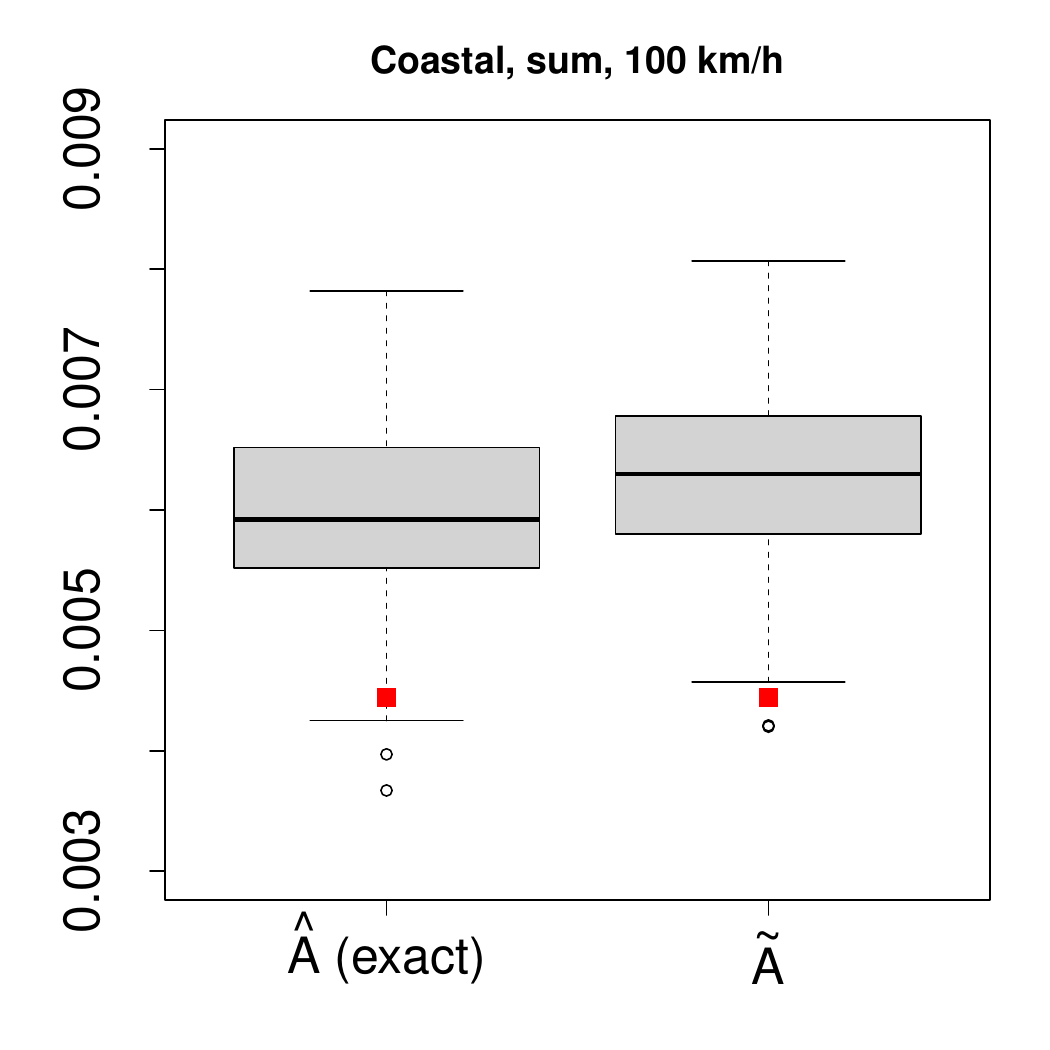}  \\
\vspace{-0.25cm}
\includegraphics[width=0.33\textwidth]{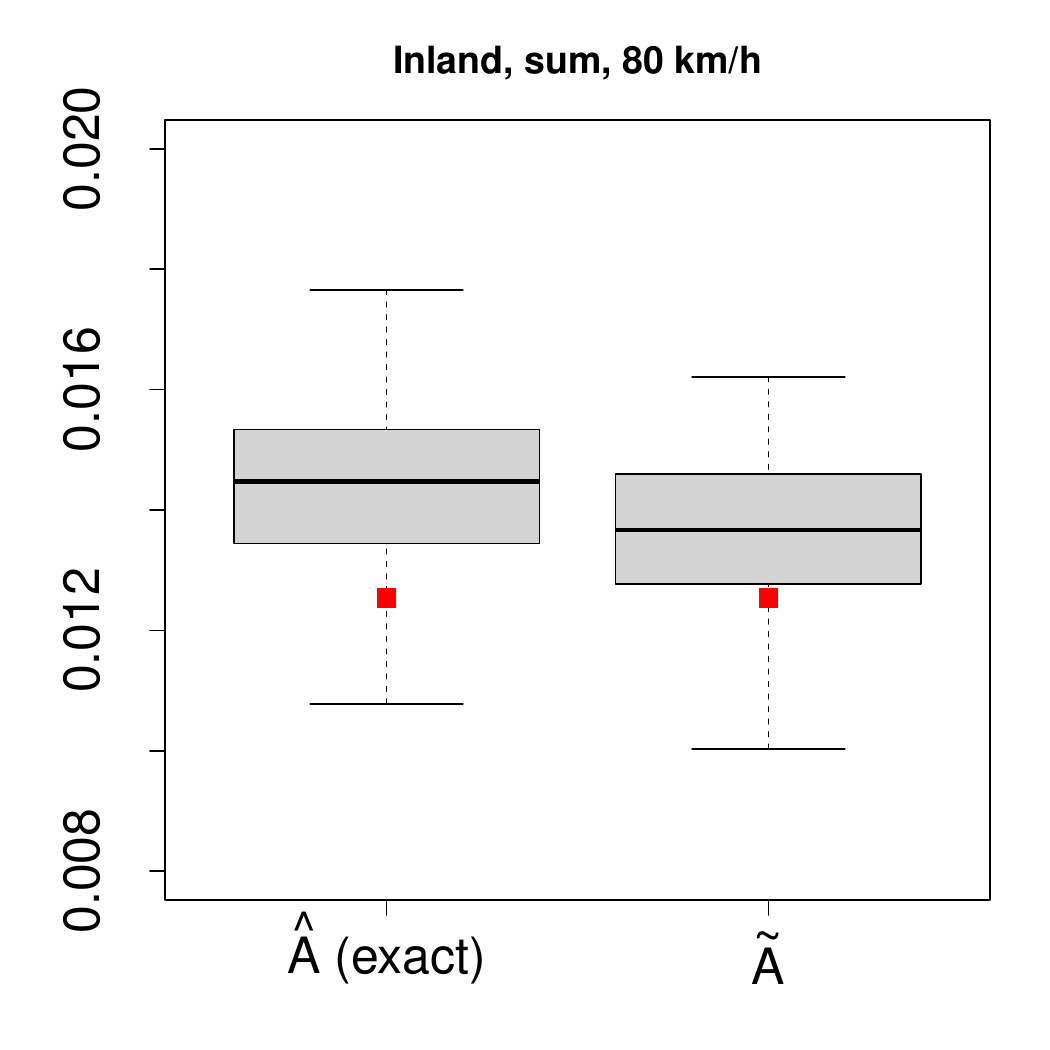} \hspace{-0.25cm}
\includegraphics[width=0.33\textwidth]{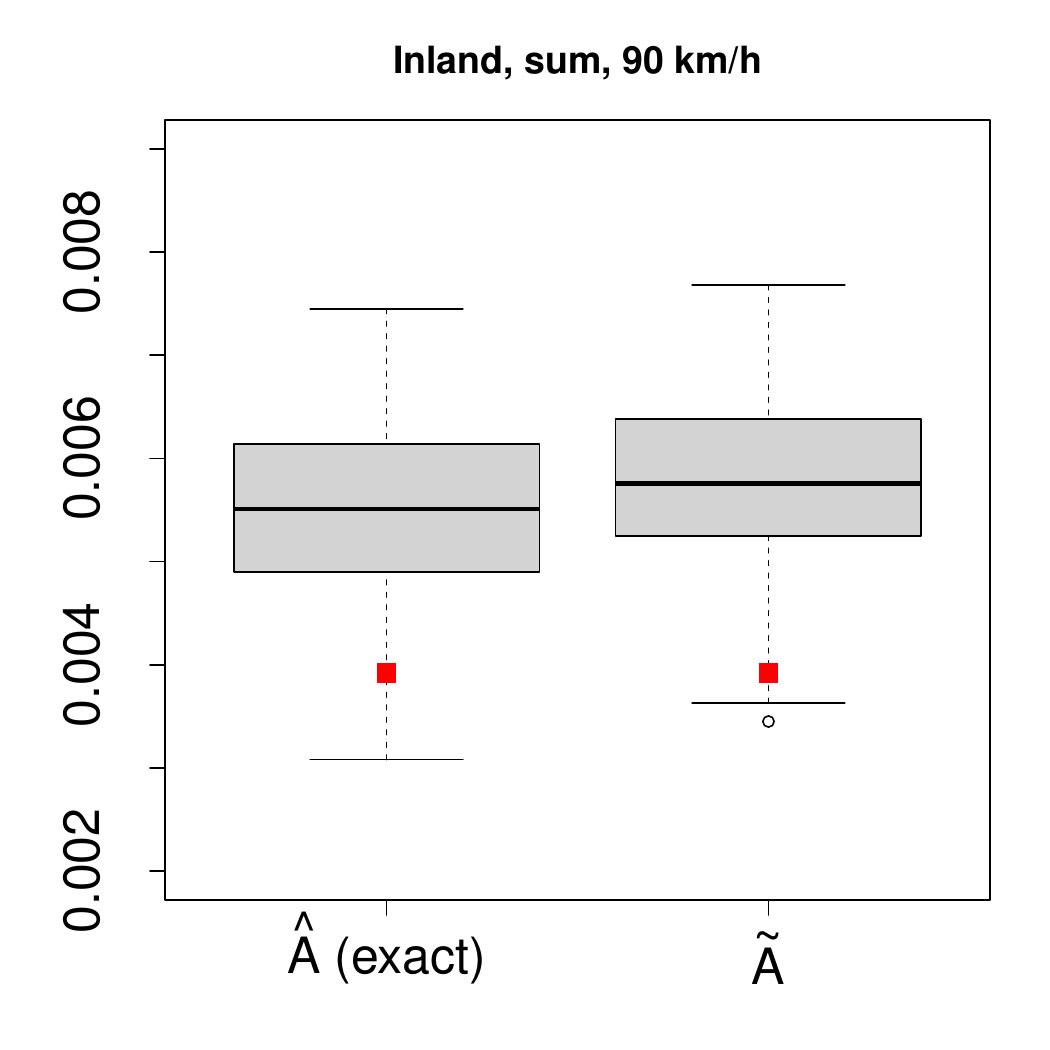}  \hspace{-0.25cm}
\includegraphics[width=0.33\textwidth]{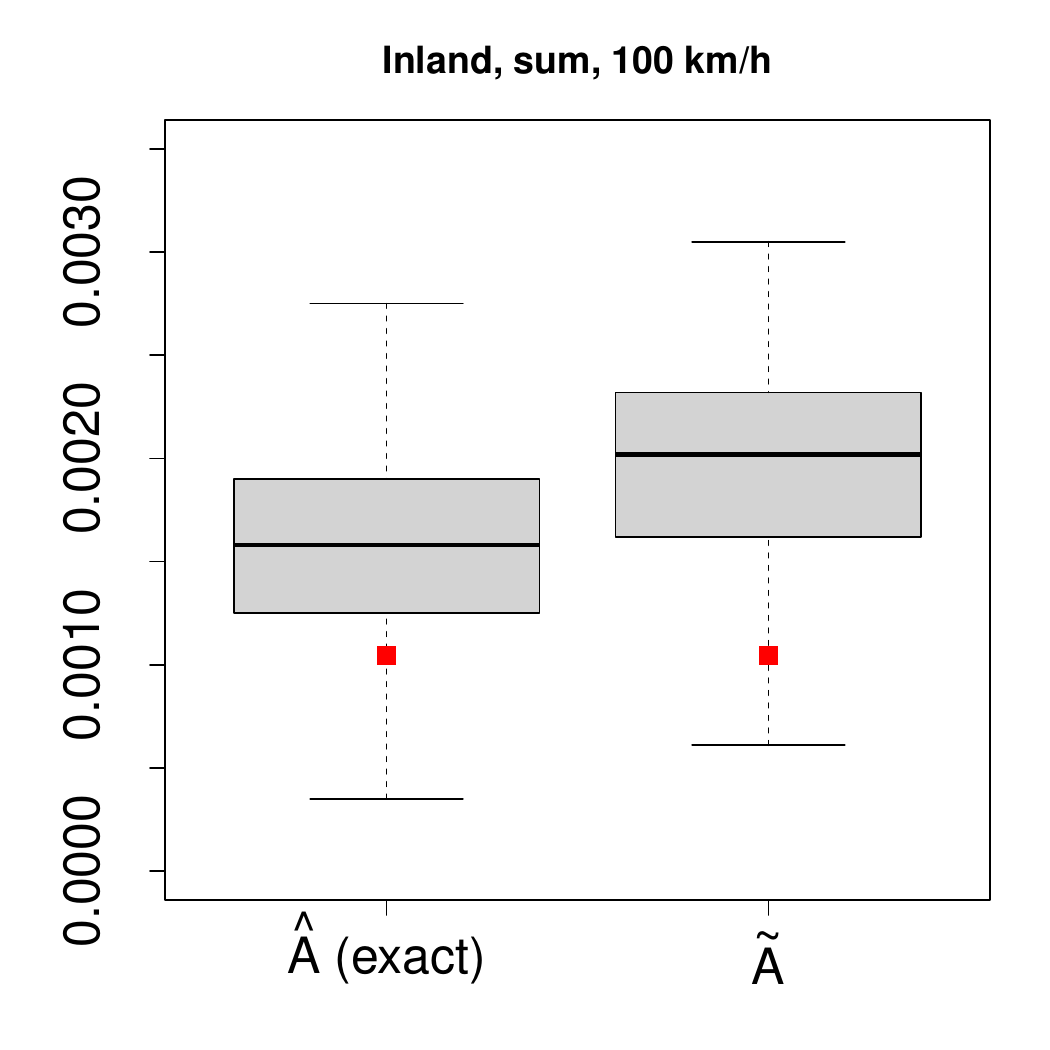}  
\caption{Boxplots of estimates of failure probabilities for the regions $C_{\textnormal{max}}$ (first and second rows) and $C_{\textnormal{sum}}$ (third and fourth rows) for the 17 coastal stations (first and third rows) and the 18 inland stations (second and fourth rows). Red dots correspond to empirical estimates.}
\label{fig:wind3}
\end{figure}

\section{Discussion}\label{sec:disc}
\change{We proposed a general method for the estimation of failure probabilities based on pairwise measures of tail dependence summarized in the TPDM. Our approach is fast and convenient for moderately large $d$, e.g., $d\leq 30$. For higher dimensions, exact decompositions become increasingly difficult to obtain. As a consequence, the diagonal elements of the TPDM (i.e., the asymptotic scales of the variables) may be overestimated, which restricts the applicability of the approach to certain classes of failure regions.

Summarizing the tail dependence structure through bivariate measures inevitably entails a loss of information. Nevertheless, our numerical experiments demonstrate that this loss has a limited impact on the estimated failure probabilities. In particular, the application in Section~\ref{sec:wind} shows that the proposed method remains effective even when the data are not multivariate regularly varying and when the failure region becomes intractable after marginal standardization.

An additional source of uncertainty arises from the estimation of the tail index $\alpha$ and the tail pairwise dependence matrix $\Sigma$. Although a multitude of estimators exist for $\alpha$, estimation of $\Sigma$ has yet to be investigated further and a better estimator could potentially improve the accuracy of the proposed failure probability estimates.}

\section{Disclosure statement}\label{disclosure-statement}

The authors have no conflicts of interest to declare.

\section{Data Availability Statement}\label{data-availability-statement}

All data and code are available from \url{https://github.com/akiriliouk/FailureProbTPDM}.

\phantomsection\label{supplementary-material}
\bigskip

\begin{center}

{\large\bf SUPPLEMENTARY MATERIAL}

\end{center}

The supplement contains all proofs and an additional data application for heavy rainfall in Switzerland.

\bgroup 
\singlespacing

  \bibliography{reference.bib}

\egroup

\end{document}